%% file: main.tex
\documentclass[sigconf,nonacm]{acmart}
\AtBeginDocument{%
  }
\usepackage{url}

\usepackage{breakurl}
\usepackage{hyperref}
\usepackage{subcaption}
\usepackage{textcomp}
\usepackage{xcolor}
\usepackage{glossaries}
\usepackage{svg}
\usepackage{algorithm}
\usepackage{multirow}
\input{acronyms}

\usepackage{graphicx}
\usepackage{booktabs,arydshln}
\usepackage{pdflscape}
\usepackage{subcaption}
\usepackage{ragged2e}
\usepackage{algorithm}
\usepackage{algpseudocode}
\usepackage{cleveref}
\usepackage{siunitx}
\usepackage{tikz}
\usepackage{enumitem}
\usepackage{adjustbox}
\usepackage{balance}
\makeatletter
\def\adl@drawiv#1#2#3{%
        \hskip.5\tabcolsep
        \xleaders#3{#2.5\@tempdimb #1{1}#2.5\@tempdimb}%
                #2\z@ plus1fil minus1fil\relax
        \hskip.5\tabcolsep}
\newcommand{\cdashlinelr}[1]{%
  \noalign{\vskip\aboverulesep
           \global\let\@dashdrawstore\adl@draw
           \global\let\adl@draw\adl@drawiv}
  \cdashline{#1}
  \noalign{\global\let\adl@draw\@dashdrawstore
           \vskip\belowrulesep}}
\makeatother
\glsdisablehyper

\setcopyright{acmlicensed}
\copyrightyear{2018}
\acmYear{2018}
\acmDOI{XXXXXXX.XXXXXXX}
\acmConference[Conference acronym 'XX]{Make sure to enter the correct
  conference title from your rights confirmation email}{June 03--05,
  2018}{Woodstock, NY}
\acmISBN{978-1-4503-XXXX-X/2018/06}

\begin{document}

\title{Cross-Chain Arbitrage: The Next Frontier of MEV in Decentralized Finance}

\author{Burak Öz}
\affiliation{%
  \institution{Technical University of Munich \& Flashbots}
  \city{Munich}
  \country{Germany}}
\email{burak.oez@tum.de}

\author{Christof Ferreira Torres}
\affiliation{%
  \institution{Instituto Superior Técnico \& University of Lisbon \& INESC-ID}
  \city{Lisbon}
  \country{Portugal}}
\email{christof.torres@tecnico.ulisboa.pt}

\author{Christoph Schlegel}
\affiliation{%
  \institution{Flashbots}
  \city{Zurich}
  \country{Switzerland}}
\email{christoph@flashbots.net}

\author{Bruno Mazorra}
\affiliation{%
  \institution{Flashbots}
  \city{Barcelona}
  \country{Spain}}
\email{bruno@flashbots.net}

\author{Jonas Gebele}
\affiliation{%
  \institution{Technical University of Munich}
  \city{Munich}
  \country{Germany}}
\email{jonas.gebele@tum.de}

\author{Filip Rezabek}
\affiliation{%
  \institution{Technical University of Munich}
  \city{Munich}
  \country{Germany}}
\email{rezabek@net.in.tum.de}

\author{Florian Matthes}
\affiliation{%
  \institution{Technical University of Munich}
  \city{Munich}
  \country{Germany}}
\email{matthes@tum.de}

\renewcommand{\shortauthors}{Öz et al.}

\begin{abstract}
Decentralized finance (DeFi) markets spread across Layer-1 (L1) and Layer-2 (L2) blockchains rely on arbitrage to keep prices aligned. Today most price gaps are closed against centralized exchanges (CEXes), whose deep liquidity and fast execution make them the primary venue for price discovery. As trading volume migrates on-chain, cross-chain arbitrage between decentralized exchanges (DEXes) will become the canonical mechanism for price alignment. Yet, despite its importance to DeFi—and the on-chain transparency making real activity tractable in a way CEX-to-DEX arbitrage is not—existing research remains confined to conceptual overviews and hypothetical opportunity analyses.

We study cross-chain arbitrage with a profit-cost model and a year-long measurement. The model shows that opportunity frequency, bridging time, and token depreciation determine whether inventory- or bridge-based execution is more profitable. Empirically, we analyze one year of transactions (September 2023 - August 2024) across nine blockchains and identify 242,535 executed arbitrages totaling 868.64 million USD volume. Activity clusters on Ethereum-centric L1-L2 pairs, grows 5.5x over the study period, and surges—higher volume, more trades, lower fees—after the Dencun upgrade (March 13, 2024). Most trades use pre-positioned inventory (66.96\%) and settle in 9s, whereas bridge-based arbitrages take 242s, underscoring the latency cost of today's bridges. Market concentration is high: the five largest addresses execute more than half of all trades, and one alone captures almost 40\% of daily volume post-Dencun. We conclude that cross-chain arbitrage fosters vertical integration, centralizing sequencing infrastructure and economic power and thereby exacerbating censorship, liveness, and finality risks; decentralizing block building and lowering entry barriers are critical to countering these threats.
\end{abstract}

%

%
\keywords{Cross-Chain Arbitrage, Decentralized Finance, Maximal Extractable Value, Blockchain Interoperability}

    
\maketitle

\section{Introduction}
The blockchain ecosystem is inherently multi-chain. \gls{l1} networks such as Ethereum supply security and decentralization, while a growing number of \glspl{l2} deliver cheaper, higher-throughput execution. Nearly every chain now hosts its own \gls{defi} markets, collectively processing several billion USD in daily trading volume \cite{coingecko}. This market fragmentation, however, means prices often diverge across chains. Restoring parity—and thus market efficiency—relies on arbitrage, a major form of \gls{mev} \cite{daian_flash_2020} in which traders buy low on one market and sell high on another.

Today, most price gaps are closed by arbitraging on-chain \glspl{dex} against off-chain \glspl{cex} \cite{heimbach_non-atomic_2024}, whose deep liquidity, low fees, and fast execution—advantages afforded by their centralized infrastructure—make them the primary venue for price discovery. As blockchain execution improves, \gls{defi} adoption grows, and long-tail tokens (which can be issued permissionlessly on-chain) remain unavailable on \glspl{cex}, trading volume is expected to shift to \glspl{dex}—a long standing goal of the \gls{defi} industry. In that on-chain future, \emph{cross-chain \gls{dex}-to-\gls{dex} arbitrage} will be the canonical mechanism for price alignment.

Cross-chain arbitrage can be executed in two main ways: \textbf{(i)} by keeping inventory on multiple chains or \textbf{(ii)} by moving assets through a bridge. Holding inventory ties up capital and exposes the trader to price swings on each chain but allows near-instant execution when an opportunity arises. Bridging avoids those inventory risks yet incurs transfer delays, exposing the opportunity to competitors who can act first or to routine trading activity that can close the price gap. Arbitrageurs must therefore choose a method based on the pairs they target and the bridges available. \Cref{fig:arb-strats} illustrates real-world examples of both methods (details in \Cref{examples}).

\begin{figure*}[t!]
    \centering
    \begin{subfigure}[t]{0.48\textwidth}
        \centering
        \includegraphics[width=\linewidth]{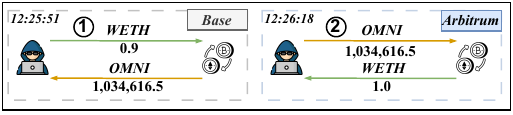}
        \caption{Inventory Arbitrage.}
        \label{fig:genaral-inv-arb}
    \end{subfigure}
    \hfill
    \begin{subfigure}[t]{0.48\textwidth}
        \centering
        \includegraphics[width=\linewidth]{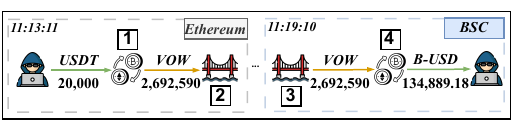}
        \caption{Bridge Arbitrage.}
        \label{fig:genaral-bridge-arb}
    \end{subfigure}
    \caption{Cross-chain arbitrage execution methods.  
    (\subref{fig:genaral-inv-arb}) Inventory arbitrage: the trader holds capital on both chains, so both legs can settle almost simultaneously. 
    (\subref{fig:genaral-bridge-arb}) Bridge arbitrage: the trader moves assets across chains between legs, and the bridge delay dominates total settlement time.
    Full transaction details appear in \Cref{examples}.}
    \label{fig:arb-strats}
\end{figure*}


Despite its importance to \gls{defi}—and the on-chain transparency that lets us track dynamics invisible in \gls{cex}-to-\gls{dex} arbitrage—cross-chain arbitrage remains under-explored. While single-domain \gls{mev} is well documented \cite{daian_flash_2020,qin2022quantifying,yang_countermeasures,schwarz-schilling_time_2023,oz_time_2023, rollinginshadows_torres, frontrunner_jones_torres_2021}, work on cross-chain arbitrage, a subset of cross-domain \gls{mev}, is sparse. Conceptual overviews of cross-domain \gls{mev} lack a domain-specific model that contrasts inventory and bridge arbitrage \cite{obadia_unity_2021, mcmenamin2023sokcrossdomainmev}. Empirical work has either examined \gls{cex}-to-\gls{dex} arbitrage by observing only the on-chain leg \cite{heimbach_non-atomic_2024} or studied hypothetical, inventory-based cross-chain opportunities limited to a few market pairs \cite{gogol_cross-rollup_2024, mazor_empirical_2023}.

We close this gap by combining a profit-cost model for cross-chain arbitrage—which determines when to bridge versus hold inventory—with a year-long empirical study on executed arbitrages across nine blockchains. 

\noindent
\textbf{Contributions}. We make the following contributions:
\begin{itemize}
   \item We present a profit-cost model that compares inventory and bridge cross-chain arbitrage, accounting for the costs of non-instantaneous bridging and holding inventory. The model solution reveals opportunity frequency, bridging time, and token depreciation as key determinants.
   \item We develop a general methodology for (i) detecting executed cross-chain arbitrages and (ii) classifying the bridge transactions that link their two legs.
  \item We conduct the first large-scale study of executed cross-chain arbitrages—\SI{242535}{} trades totaling \SI{868.64}{million} USD volume across nine blockchains—revealing activity concentrated on Ethereum-centric \gls{l1}-\gls{l2} pairs, growing by 5.5$\times$ over the study period and spiking (higher volume, more trades, lower fees) after the Dencun upgrade (March 13, 2024).
  \item We show that most arbitrage (\SI{66.96}{\%}) relies on pre-positioned inventory rather than bridges, settling in a \SI{9}{\second} median, whereas bridge-based trades take \SI{242}{\second}, highlighting the latency cost of today's bridges.
  \item We uncover rising market concentration: five largest addresses generate over half of all trades, and one alone (\texttt{0xCA74}) captures up to \SI{40}{\%} of daily volume in the post-Dencun period.
  \item We point that cross-chain arbitrage encourages vertical integration, concentrating sequencing infrastructure and economic power and thereby heightening censorship, liveness, and finality risks; we outline decentralizing block building and lowering entry barriers as countermeasures.
\end{itemize}

\section{Background}
This section provides relevant background on \gls{defi}, \gls{mev}, \gls{l1} and scaling solutions, and blockchain interoperability and bridges. 

\subsection{Decentralized Finance}
\gls{defi} offers similar primitives as \gls{tradfi} but extends it via additional solutions that are solely enabled by blockchain technology.
These include \glspl{dex}, lending platforms, options markets, and tokenized assets, all implemented via smart contracts. 
While \glspl{cex} use a \gls{clob} model, which relies on a centralized entity for order matching and settlement, \glspl{dex} typically adopt \glspl{amm} to facilitate \gls{p2p} trading, although the \gls{clob} model has also been used historically (e.g., EtherDelta). 

\glspl{amm} rely on liquidity pools and deterministic pricing rules to execute trades without intermediaries. A widely used design is the \gls{cpmm} adopted by Uniswap V2 \cite{univ2}: for pool reserves \(x\) and \(y\), all swaps must satisfy \(x \cdot y = k\). A trade moves the reserves along this curve, setting a new price, while liquidity
deposits or withdrawals change the curve (and hence \(k\)) without altering the instantaneous price.

Trading on \glspl{dex} therefore introduces price slippage—both expected, due to trade volume and liquidity constraints, and unexpected, due to execution delays and market volatility. Either form can push \gls{cpmm} prices out of line with other pools or venues, creating the arbitrage opportunities that traders exploit.

\subsection{Maximal Extractable Value}
The decentralized infrastructure enables opportunities for \gls{mev} extraction.
\gls{mev} relies on two key transaction ordering primitives: frontrunning, where an extractor ensures their transaction precedes a target transaction (\(T_{Target}\)), and backrunning, where the extractor's transaction executes immediately after \(T_{Target}\). 
Common \gls{mev} strategies include arbitrage, liquidation, and sandwiching. 
Arbitrage exploits price discrepancies across exchanges by analyzing blockchain state changes, ensuring price alignment across \glspl{dex}. 
Liquidations involve repaying a debt to purchase discounted collateral, often focusing on fixed-discount opportunities that can be executed in a single transaction. 
Both arbitrage and liquidation are generally considered beneficial for market efficiency. 
In contrast, sandwiching is a manipulative strategy where an adversarial transaction wraps a target trade (\(T_V\)), buying the asset beforehand to profit from the price increase caused by \(T_V\), and subsequently selling at a higher price. 
This practice disrupts fair-price execution and is considered harmful. 

Competition among \gls{mev} extractors may cause block congestion and gas fee inflation. 
\gls{mev} can also result in systemic instability by incentivizing behaviors prioritizing individual gain over network security. 
In the context of blockchain, this can lead to consensus instabilities, as miners or validators may be incentivized to steal transactions from others by creating forks, thereby undermining the system's integrity. 
Among the various strategies for \gls{mev} extraction, arbitrage
is the most prevalent, as demonstrated in measurement studies in \cite{daian_flash_2020,qin2022quantifying, rollinginshadows_torres}.

While \gls{mev} is typically studied on a single chain, value can also be extracted across domains.  A domain is any system with a shared mutable state—\gls{l1} and \gls{l2} blockchains, \glspl{cex}.  \cite{obadia_unity_2021} formalize cross-domain \gls{mev} as the maximum cumulative balance increase a user can achieve by controlling transaction sequencing across several such domains. In this work we focus on one concrete subset: executed \gls{dex}-to-\gls{dex} arbitrage that spans multiple blockchains.

\subsection{Layer-1 and Scaling Solutions}
\gls{l1} refers to the classical blockchain model where transactions are recorded on a public, immutable, and trustless ledger which follows a consensus algorithm to determine block creation. 
Examples of popular \glspl{l1} include Ethereum \cite{wood2014ethereum}, Binance Smart Chain \cite{bnbchain2025Jan}, Avalanche \cite{rocket2020scalableprobabilisticleaderlessbft}, etc. 
However, \glspl{l1} can face scalability issues in terms of throughput and transaction fees. 
As a result, a number of so-called \gls{l2} scaling solutions have emerged, where the most prominent solutions are either based on sidechains (e.g., Polygon \cite{polygon}) or commit chains (e.g., rollups).
Commit chains or rollups can be further split into either optimistic (e.g., Arbitrum \cite{arbitrum-docs}, Base \cite{base-docs}, Optimism \cite{op-docs}) or zero-knowledge based (e.g., ZKsync \cite{zksync-docs}, Scroll \cite{scroll2023whitepaper}) depending on what types of proofs are used to verify the validity of \gls{l2} transactions.
Sidechains typically run a fast consensus mechanism among few peers in parallel to \gls{l1}.
On the other hand, rollups enable throughput scaling by off-loading compute and (possibly) storage resources off-chain without a need for large-scale consensus, which is provided by the underlying \gls{l1}.
Generally, \gls{l2} scaling solutions allow distrustful parties to deposit funds into a bridge smart contract on \gls{l1} and then operate on \gls{l2} via \gls{l2} transactions whose state is then updated to \gls{l1}.

\Cref{tab:blockchains_overview} summarizes the various blockchains that we analyze for cross-chain arbitrage. 
Unlike \glspl{l1}, which typically operate using a gas price model and provide public access to mempool data, \glspl{l2} typically feature private mempools with transactions ordered by centralized sequencers, which often follow a \gls{fcfs} strategy.
The latter makes it harder for users to extract \gls{mev} as they cannot observe other pending transactions nor can pay higher gas fees to prioritize their own transactions.
The block time does not correspond to the finality, as \glspl{l2} only offer a soft finality to its users as the transactions are only fully finalized after being settled on \gls{l1}.

\input{tables/blockchains_overview}

\subsection{Blockchain Interoperability and Bridges}

Blockchain interoperability aims to connect blockchain networks and transfer assets or data between them. The transfer typically involves locking assets on one chain and minting a matching representation on the destination chain, conditional on some evidence of the asset being locked. Bridges are the basic infrastructure enabling this interoperability.
Using a bridge, a token can be moved from one blockchain to another while maintaining some security guarantees, with the form of guarantee depending on the means through which the locking of funds is proved to minting contracts.
Bridges can be categorized into either \textit{native} or \textit{multi-chain}.
Each solution comes with different security, latency, and cost trade-offs.

Native bridges are built into the architecture of the underlying blockchain and facilitate direct transfers between their \glspl{l1} and \glspl{l2}. 
The process typically involves asset locking via a smart contract, proof generation for the locked asset, and transmission of the asset to the destination chain.
On the destination chain, an equivalent amount of the asset is either minted (if the bridge uses wrapped tokens) or unlocked.
These bridges often prioritize security and are tightly integrated with the chain's architecture.
Multi-chain bridges, on the other hand, are more versatile as they can support several different blockchains, while native bridges only operate between two blockchains. However, such bridges are operated by third-party companies, introducing trust assumptions for counter-party risk. 
Multi-chain bridges follow the same steps of locking, proof generation, and minting/unlocking on the destination chain, except that they present a unified communication mechanism that is blockchain-agnostic.


\section{Theoretical Analysis}

In this section, we model cross-chain arbitrage and solve it to provide theoretical results on the profitability of inventory against bridge arbitrages, and the factors impacting this decision.

\subsection{Model}

We model trading across two chains. On the first chain, trading happens on a \gls{cpmm} with reserves $R^A_t$ and $R^B_t$ of the two tokens where the subscript denotes potential time dependency. On the second chain, trading happens on a perfectly liquid market where arbitrary amounts of tokens can be exchanged at a token $A$ to token $B$ exchange rate of $Q_t$.  This model is a good approximation to reality, even in the case that the second market also uses a \gls{cpmm} provided that the second market has more liquidity than the first market. We choose a perfectly liquid market for ease of exposition, to make the cost of inventory calculation less involved. See, however, \Cref{coi_bounded} for a derivation of the cost of inventory in the case of a \gls{cpmm} with limited liquidity in the second market.~\footnote{Qualitatively, the results would be similar, with the caveat that liquidity constraints in the second market make bridging relatively more attractive than buying and maintaining inventory in the second market, which would be more costly in that case.}

Token $A$ is the numéraire for the subsequent calculations, i.e., we measure value and cost in token $A$.
We assume that the two market, and the bridge operate without fees, and that gas and transaction costs for swaps and bridging are negligible.
 Moreover, we assume that the exchange rate in the second market follows a geometric Brownian motion with percentage drift $\mu$ and percentage variance $\sigma$. We assume that the risk-free rate was already subtracted so that $\mu$ is the excess return over the risk-free investment. We generically think of token $B$ as a token with negative excess return over the risk-free investment so that $\mu<0$. Implicitly, this means that the market is incomplete so that token $B$ cannot, or only at high cost, be shorted. If on the other hand, $\mu> 0$, the subsequent discussion would be somehow trivial, as in that case the arbitrageur would rather want to hold token $B$ permanently after having bought it, rather than selling it for arbitrage profits on the second market.\footnote{Alternatively, we could also have non-negative excess return, but make the trader risk averse and obtain qualitatively similar results.}

We assume that arbitrage opportunities of equal size arrive according to a Poisson process with rate $\lambda$: this means that the exchange rates $P_t:=R^A_t/R_t^B$ and $Q_t$ in the two markets are the same almost always, but every 
$\lambda^{-1}$ minutes on average, we have $P_\tau<Q_\tau=pP_\tau$ for a constant factor $p>1$, where $\tau$ is the random arrival time of the opportunity.\footnote{The equal size assumption on arbitrage opportunities is for tractability of the model.} In a hypothetical world without frictions where tokens can be bridged instantaneously between the two market at zero cost or the arbitrageur can hold inventory of the second token in the second market at zero cost, a profit-maximizing arbitrageur sells an amount of
$$(\sqrt{p}-1)R_\tau^A$$
token $A$ in the first market
to obtain an amount of $$(1-\sqrt{1/p})R_\tau^B$$
 token $B$ that he sells in the second market to obtain an overall profit of
$$(\sqrt{p}-1)^2R_\tau^A.$$
\Cref{fig:arb_prof_model} plots the change in arbitrage profit with respect to $p$. See, e.g., \cite{fritsch2024mevcapturetimeadvantagedarbitrage} for a derivation of the profit and optimal trade size.

\subsubsection{Non-Instantaneous Bridging}
Now suppose, bridging is non-instantaneous and it requires time $\Delta>0$ to bridge funds from the first to the second market. If the arbitrageur does not hold inventory on the second market, and needs to bridge from the first to the second market to complete the arbitrage, then the above formula for the optimal trade size needs to be adjusted for the expected token price depreciation during the bridging. An expected profit-maximizing arbitrageur chooses his trade expecting that after bridging he can exchange the $B$ tokens at an (in expectation) worse exchange rate $Q_{t+\Delta}$ rather than the instantaneous rate $Q_t$. Thus, he chooses to trade an amount of 
$$\left(\sqrt{p}e^{\mu\Delta/2}-1\right)R_\tau^A$$
token $A$,
to make an expected profit\footnote{The expectation is at time $\tau$ (i.e., conditional on the realization of $\tau$, $R_\tau$ and $Q_{\tau}$) over the uncertain realization of $Q_{\tau+\Delta}$. Subsequently, we denote by subscripts the time at which we take expectation.} of 

\begin{figure}[t!]
    \centering
    \includegraphics[width=\linewidth]{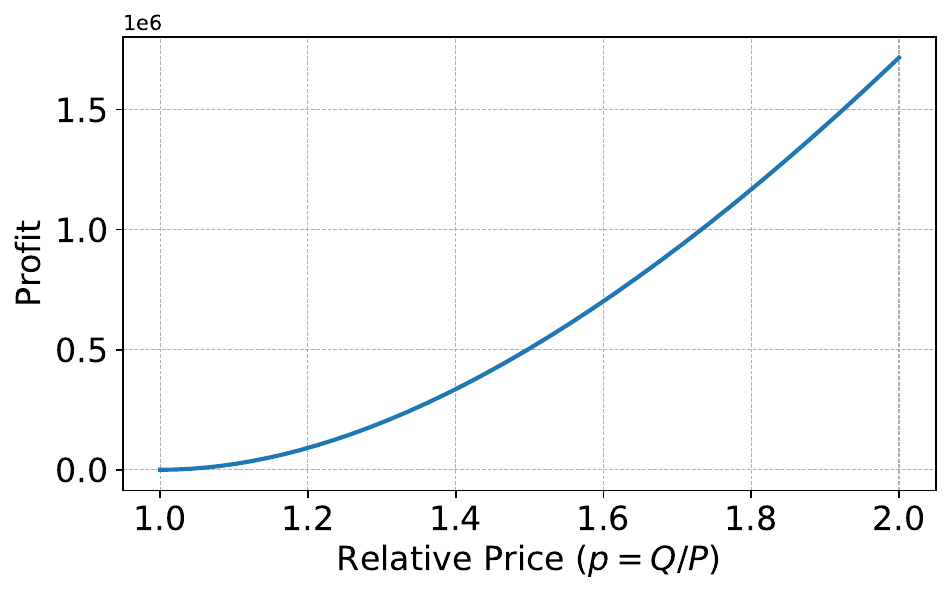}
    \caption{Arbitrage profit as a function of the relative price (\(p = Q/P\)) between a \gls{cpmm} (\(P\)) and a perfectly liquid market (\(Q\)). Token reserve of the \gls{cpmm} is fixed at \(R=10^{7}\).}
\label{fig:arb_prof_model}
\end{figure}

\begin{align*}&E_{\tau}[Q_{\tau+\Delta}\left(1-\sqrt{p^{-1}}\right)R_\tau^B-\left(\sqrt{p}-1\right)R_\tau^A]\\=&E_{\tau}[\frac{Q_{\tau+\Delta}}{Q_{\tau}}]\left(1-\sqrt{p^{-1}}e^{-\mu\Delta/2}\right)R_\tau^A-\left(\sqrt{p}e^{\mu\Delta/2}-1\right)R_\tau^A\\=&\left(\sqrt{p}e^{\mu\Delta/2}-1\right)^2R_\tau^A.\end{align*}

Note that the expected profit is strictly smaller than in the frictionless case, as $\Delta>0.$  We call the difference in profit per unit of liquidity in the \gls{amm},
$$C^{BR}:=(\sqrt{p}-1)^2-(\sqrt{p}e^{\mu\Delta/2}-1)^2,$$


the marginal \emph{cost of non-instantaneous bridging}. \Cref{fig:bridge_cost_model} illustrates the change in this cost depending on bridging time.

\subsubsection{Costly Inventory}
Next, we compare this to the scenario where the arbitrageur holds inventory of token $B$ in the second market. In this case, the arbitrageur optimally needs to maintain an inventory of $I_t:=(1-\sqrt{p^{-1}})R_t^B$  units of $B$ tokens on the second market.

The cost of inventory is given by the following expression using a standard calculation:
\begin{lemma}
The expected cost of inventory is 
\begin{align*}
C(I)&:=-E_0[\int_0^{\tau}I_tdQ_t]
\end{align*}
\end{lemma}
\begin{proof}
Let the value of the portfolio of the arbitrageur at time $t$ be $V_t\equiv x_t+Q_tI_t$ where $x_t$ are his token $A$ holdings. We use a discrete approximation and then pass to the limit:
Let the time between $0$ and the random arrival time $\tau$ be partitioned by $0=t_0<t_1<\ldots<t_n=\tau.$  If the arbitrageur re-adjusts his portfolio at discrete times $t_0,\ldots,t_n$, his token $A$ holding change from period to period by $x_{t_i}-x_{t_{i-1}}=Q_{t_i}(I_{t_{i-1}}-I_{t_i})$. Thus, the portfolio value between re-adjustments changes by 
\begin{align*}V^{(n)}_{t_i}-V^{(n)}_{t_{i-1}}=Q_{t_i}I_{t_i}+x_{t_i}-Q_{t_{i-1}}I_{t_{i-1}}-x_{t_{i-1}}\\=Q_{t_i}I_{t_i}-Q_{t_{i-1}}I_{t_{i-1}}+Q_{t_{i}}(I_{t_{i-1}}-I_{t_{i}})=(Q_{t_i}-Q_{t_{i-1}})I_{t_{i-1}}.
\end{align*}Thus the total change in the value of his inventory over time is
\begin{align*}V^{(n)}_0-V^{(n)}_{\tau}=\sum_{i=1}^n (V^{(n)}_{t_{i-1}}-V^{(n)}_{t_i})=\sum_{i=1}^n(Q_{t_{i-1}}-Q_{t_i})I_{t_{i-1}}\end{align*}
which in the limit for finer and finer partitions becomes an Ito integral
$$V_0-V_\tau=\lim_{n\to\infty}V^{n}_0-V^{n}_\tau=-\int_0^\tau I_tdQ_t. $$
Taking expectations yields the result.
\end{proof}

\begin{figure}[t!]
    \centering
    \includegraphics[width=\linewidth]{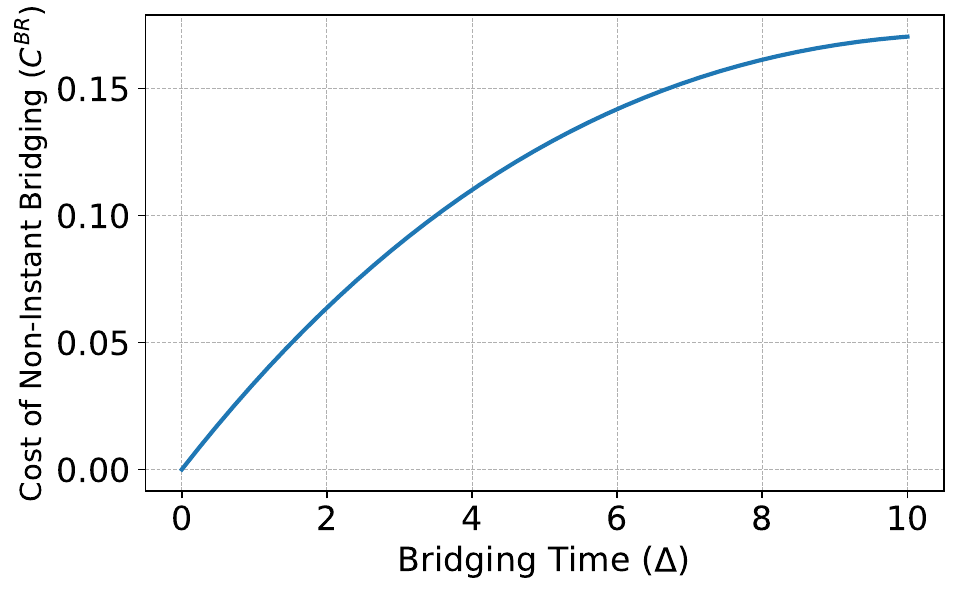}
    \caption{Cost of non-instantaneous bridging (\(C^{BR}\)) as a function of the bridging time (\(\Delta\)). Parameters: relative price $p = 2$, drift $\mu = -6.25\%$.}
\label{fig:bridge_cost_model}
\end{figure}

The cost will depend on the evolution of liquidity in the first market over time, as this determines the size of the optimal arbitrage trade. In the following, we consider the special case where the reserves take the form
$$ R_t^B=(\frac{Q_0}{Q_t})^kR_0^B$$
which gives the cost per unit traded a particularly simple form, as shown below. While the cost would take a different form in general, the class contains two reasonable and tractable cases as special case:
\begin{enumerate}
\item Liquidity in- and outflows in the \gls{amm} are such that the value of the inventory in the \gls{cpmm} (in terms of token $A$), $R_t^A+Q_tR_t^B=2R_t^A$,  is constant. This corresponds to the case $k=0$. This is a reasonable approximation if the time horizon considered is relatively short and the market is sufficiently liquid.
\item There are no inflows of liquidity so that $R_t^AR_t^B=R^A_0R_0^B$ for all $t>0$. This corresponds to the case $k=1/2$. This is a reasonable approximation if we are worried about an increasingly illiquid market so that the arbitrageur can make less profit from arbitraging over time. 
\end{enumerate}

To calculate the cost of inventory per unit traded, first note that the expected value (at time $0$) of the inventory at time $0\leq t\leq \tau$ is
$$E_0[Q_t I_t]=(1-\sqrt{\tfrac{1}{p}})R_0^AE_0[\frac{Q_0^{k-1}}{Q_t^{k-1}}]=(1-\sqrt{\tfrac{1}{p}})R_0^Ae^{(k-1)(\tfrac{1}{2}k\sigma^2-\mu)t}.$$
The expected cost (at time $0$, i.e., with respect to the uncertain realizations of the arrival time $\tau$ and of the exchange rate changes from $0$ until the arrival time $\tau$) is given by
\begin{align*}C(I)&=-E_0[\int_0^\tau I_tdQ_t]=-E_0[\int_0^\tau \mu Q_tI_tdt+\sigma I_tdW_t]\\&=-\mu(1-\sqrt{\tfrac{1}{p}})R_0^A E
[\int_0^\tau \frac{Q_0^{k-1}}{Q_t^{k-1}}]=\frac{-\mu(1-\sqrt{\tfrac{1}{p}})R_0^A}{\lambda+(1-k)(\tfrac{1}{2}k\sigma^2-\mu)}.\end{align*}
The expected value (at time $0$) of the inventory at the time of the arbitrage $\tau$ is 
$$E_0[Q_{\tau}I_\tau]=\frac{\lambda(1-\sqrt{\tfrac{1}{p}})R_0^A}{\lambda+(1-k)(\tfrac{1}{2}k\sigma^2-\mu)}.$$
The expected cost per unit traded is therefore
\begin{align*}\frac{C(I)}{E_0[Q_{\tau}I_\tau]}=\frac{-\mu}{\lambda}.
\end{align*}

The cost is strictly positive whenever $\mu<0$.



\subsection{Results}\label{theo_results}
Comparing the expected cost of bridging and the expected inventory cost at time $0$, we come to the following conclusion:
\begin{theorem}
Let the price difference in the two markets at the time of an opportunity be $p>1$. Define $$C^{BR}:=(\sqrt{p}-1)^2-(\sqrt{p}e^{\mu\Delta/2}-1)^2.$$
A profit-maximizing arbitrageur chooses to hold inventory of token $B$ if and only if
$$\frac{-\mu}{\lambda}(1-\sqrt{\tfrac{1}{p}})<C^{BR}.$$
\end{theorem}
Let us first look at the arbitrage opportunity arrival rate $\lambda$. 
The inequality is satisfied if and only if 
$$\lambda>\frac{-\mu}{C^{BR}}(1-\sqrt{\tfrac{1}{p}}).$$
Thus, if opportunities arrive frequently enough, it is worthwhile to maintain inventory.
\Cref{fig:delta_vs_lambda} plots the minimum arrival rate that makes inventory profitable over bridging for each bridging time value.

\begin{figure*}[t!]
    \centering
    \begin{subfigure}[t]{0.48\textwidth}
        \centering
        \includegraphics[width=\linewidth]{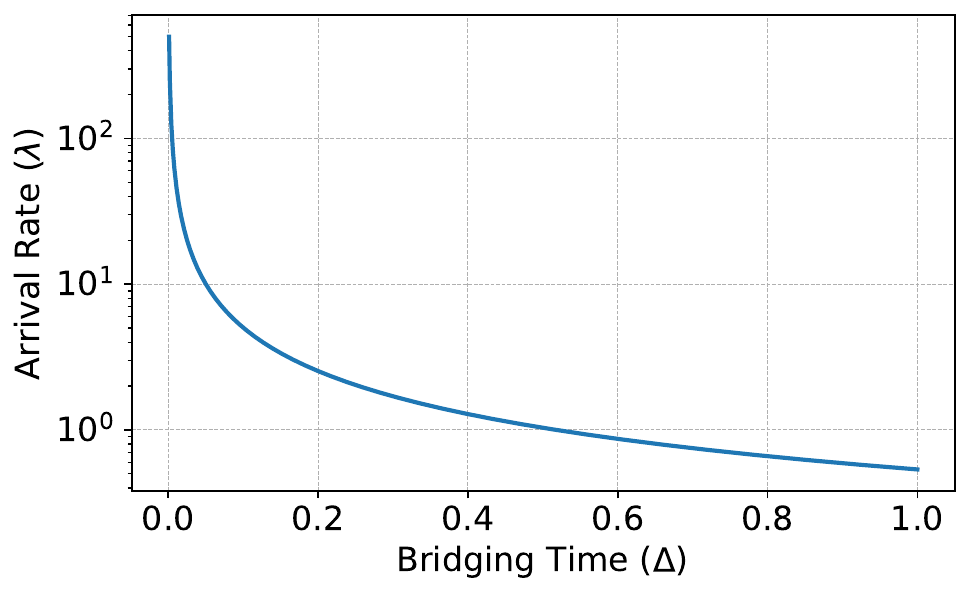}
        \caption{$\lambda$ threshold vs. $\Delta$.}
        \label{fig:delta_vs_lambda}
    \end{subfigure}
    \hfill
    \begin{subfigure}[t]{0.48\textwidth}
        \centering
        \includegraphics[width=\linewidth]{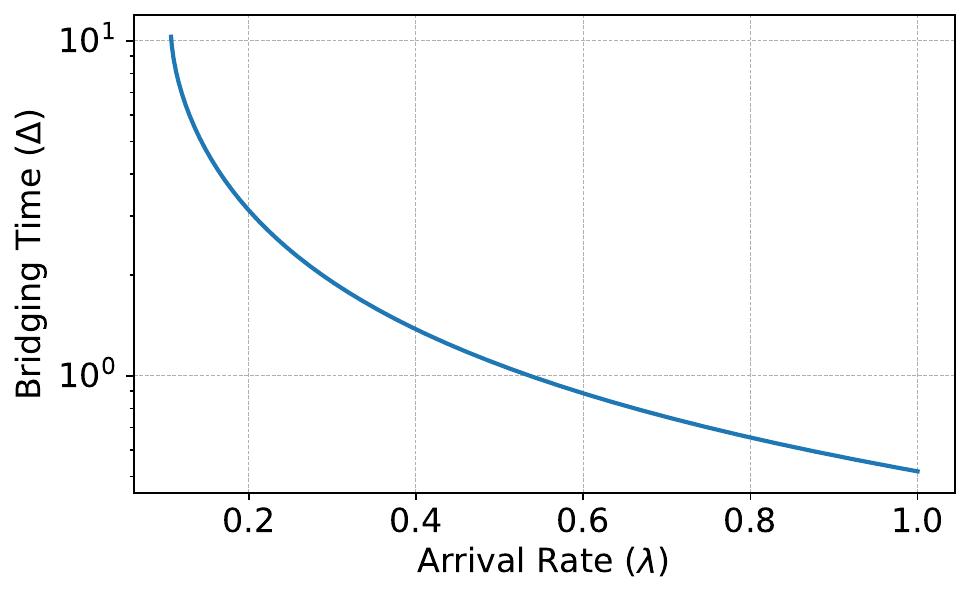}
        \caption{$\Delta$ threshold vs. $\lambda$.}
        \label{fig:lambda_vs_delta}
    \end{subfigure}
    \caption{
    Profitability boundaries for inventory and bridge arbitrage as functions of bridging time~($\Delta$) and opportunity arrival rate~($\lambda$). Parameters: relative price $p = 2$, drift $\mu = -6.25\%$. 
   (\subref{fig:delta_vs_lambda}) Minimum $\lambda$ required for inventory arbitrage to dominate as $\Delta$ varies.  
  (\subref{fig:lambda_vs_delta}) Maximum $\Delta$ tolerable for bridge arbitrage to dominate as $\lambda$ varies.
    }
    \label{fig:profit_conditions_model}
\end{figure*}


Similarly, we can look at the bridging time $\Delta$. When bridging is fast enough, it is less costly than maintaining inventory. 

 \[
 \Delta < -\frac{2}{\mu}\ln\!\left(\frac{\sqrt{p}}{1+\sqrt{M}}\right),\text{ for }M:=(\sqrt{p}-1)^2+\tfrac{\mu}{\lambda}(1-\tfrac{1}{\sqrt{p}}).
 \]

\Cref{fig:lambda_vs_delta} shows, for each arrival rate, the maximum bridging delay for which bridging still outperforms holding inventory.

For low enough bridging time we can define a threshold such that for high enough expected depreciation of the value of token $B$ the arbitrageur will bridge, while for low expected depreciation of the value of token $B$ the arbitrageur will hold inventory.

\begin{corollary}
If $1/\lambda>\Delta p$, then there is a threshold $-\lambda\leq\hat{\mu}\leq0$ such that for low expected depreciation of the value of token $B$, $\mu>\hat{\mu}$, the arbitrageur chooses to hold inventory and for high expected depreciation of the value of token $B$, $\mu<\hat{\mu}$, the arbitrageur bridges.
\end{corollary}
\begin{proof}
The expression 
$$\frac{-\mu}{\lambda}(1-\sqrt{\tfrac{1}{p}})-C^{BR}$$ is decreasing in $\mu$, as can be verified from the negativity of the derivative:
$$-\frac{1}{\lambda}(1-\sqrt{\tfrac{1}{p}})+\Delta(p e^{\mu\Delta}-\sqrt{p}e^{\mu\Delta/2})<-\frac{1}{\lambda}(1-\sqrt{\tfrac{1}{p}})+\Delta(p -\sqrt{p}).$$
\end{proof}

\Cref{fig:volatility_vs_cost} visualizes how the cost difference between inventory and bridge arbitrage changes with depreciation of the value of the token $B$, showing inventory becoming more costly below a threshold value.

\begin{figure}[b!]
    \centering
    \includegraphics[width=\linewidth]{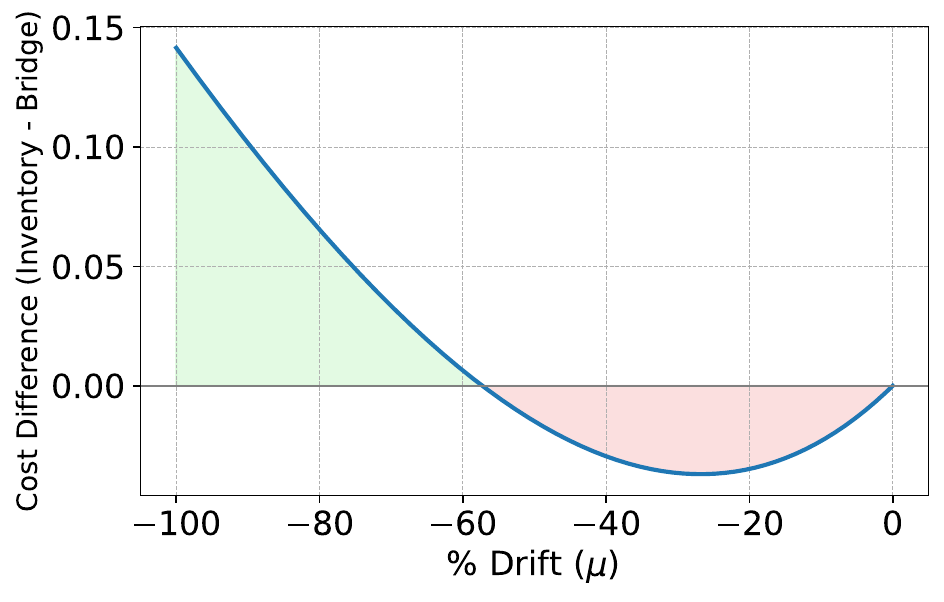}
    \caption{Cost difference (inventory – bridge) versus asset drift~($\mu$). Positive values (green) mean bridging is cheaper; negative values (red) mean holding inventory is cheaper. Parameters: relative price $p=2$, opportunity arrival rate $\lambda=1$, bridging time $\Delta=1$.}
\label{fig:volatility_vs_cost} 
\end{figure}
\section{Empirical Analysis}

In this section, we analyze a full year of executed cross-chain arbitrages across nine blockchains. After outlining our detection methodology, we describe the curated datasets and then examine the results to gauge how prevalent these arbitrages are in practice.



\subsection{Detecting Cross-Chain Arbitrage}\label{detection}
Existing arbitrage detection methods either target atomic strategies confined to a single chain \cite{qin2022quantifying, wang, mclaughlin, rollinginshadows_torres} or on non-atomic arbitrage between \glspl{cex} and \glspl{dex} \cite{heimbach_non-atomic_2024}. However, neither is sufficient for cross-chain scenarios. We therefore introduce a hybrid approach that combines their key aspects.

Similar to atomic arbitrage, we require a closed cycle of trades that starts and ends with economically equivalent assets (e.g., native, wrapped, or bridged versions of the same token), such as USDC and USDT or ETH and WETH, where each trade’s output becomes next trade's input, and the final output exceeds the initial input. Unlike atomic arbitrage, these trades span multiple transactions across different blockchains. Consequently, as with \gls{cex}-to-\gls{dex} arbitrage, single transactions may appear unprofitable in isolation but reveal profit when matched with a counterpart on another domain.

To detect such arbitrages, we focus on two-hop cross-chain cycles.  
Let \(S_i \subseteq T_i\) denote the set of \emph{swap}-event emitting transactions on blockchain \(i\) for \(i \in \{1,\dots,B\}\).  
For any transaction pair \((t_1^i,\, t_2^j)\) with \(t_1^i \in S_i\), \(t_2^j \in S_j\), and \(i \neq j\), assume each transaction has

\[
(a_{\text{in}},\,x_{\text{in}}) \longrightarrow (a_{\text{out}},\,x_{\text{out}}),
\]

where \(a\) is the asset and \(x\) the amount.  
We classify the pair as a cross-chain arbitrage if it satisfies the following heuristics:

\begin{enumerate}[label=\textbf{H\arabic*}, leftmargin=2.1em, itemsep=0.25em]
  \item Cyclic: the input and output assets of the transactions form a closed loop,
        \(a_{\text{out}}^{t_1^i}\equiv a_{\text{in}}^{t_2^j}\) and
        \(a_{\text{out}}^{t_2^j}\equiv a_{\text{in}}^{t_1^i}\).

  \item Marginal match: the intermediate amounts differ by
        at most \SI{0.5}{\%},
        \(\tfrac{\lvert x_{\text{out}}^{t_1^i}-x_{\text{in}}^{t_2^j}\rvert}
          {x_{\text{out}}^{t_1^i}}\le0.005\).

  \item Temporal window: the gap between \(t_1^i\) and \(t_2^j\)
        is \(\le\)\SI{12}{\second} for stablecoin-native pairs  
        and \(\le\)\SI{1}{\hour} otherwise.

  \item Entity link: the same \gls{eoa} submits both transactions,
        or both first interact with the same \gls{mev} contract.
\end{enumerate}

\textbf{H1} follows the standard cycle heuristic for detecting arbitrage, except that our swaps span two transactions. We match only the input and output assets of those transactions, ignoring any intermediate swap assets. The matched assets need not be identical but must be economically equivalent.

\textbf{H2} verifies that the matched swaps differ by no more than a small margin, ensuring the arbitrageur’s inventory position is unchanged aside from the final profit. For accurate accounting, we aggregate the total amounts of the original input and output assets across all swaps within a transaction. This is particularly important when trades are split into many smaller swaps, as is common with aggregator protocols that optimize execution price.

We set the marginal-difference threshold to \SI{0.5}{\%} to account for small discrepancies from bridging fees or inventory adjustments. Sampling one hour of trades on seven different days, we compared thresholds of \SI{1.0}{\%}, \SI{0.5}{\%}, and \SI{0.1}{\%}. Tightening the cut-off from \SI{1.0}{\%} to \SI{0.5}{\%} eliminated about \SI{9}{\%} false positives without losing genuine matches. Reducing it further to \SI{0.1}{\%} discarded roughly \SI{15}{\%} valid arbitrages, proving too strict. Thus, \SI{0.5}{\%} strikes an optimal balance, keeping true arbitrages while filtering spurious ones.

\textbf{H3} captures the time-sensitive nature of cross-chain arbitrage, with a focus on the temporal window between the trades. As with other strategies, arbitrageurs must execute fast to avoid being frontrun or displaced by noise trades that erase the price gap. For stablecoin-native pairs (e.g., WETH-USDC) we impose a \SI{12}{\second} limit, equal to the longest block time in our dataset (Ethereum). This strict bound is justified as such pairs are typically arbitraged against continuously running \glspl{cex} \cite{heimbach_non-atomic_2024}, making longer cross-chain execution implausible. For all other pairs we allow up to \SI{1}{\hour}, which is more of an upper bound to limit computational work. Empirically, this bound is non-restrictive: $\sim$\SI{79}{\%} of matches occur within 10 minutes, and $\sim$\SI{94}{\%} within 30 minutes, ensuring coverage without introducing spurious links.

\textbf{H4} links the two transactions to the same actor to prevent matching unrelated trades that merely share similar assets, amounts, and timing. Specifically, the two transactions must either originate from the same \gls{eoa} or both first interact with the same \gls{mev} contract. All contract addresses are cross-checked against a curated list of non-\gls{mev} contracts to avoid false positives.

After applying \textbf{H1}-\textbf{H4}, some transactions still appear in multiple candidate pairs. We de-duplicate using a three-level ranking:

\begin{enumerate}
    \item  \emph{Marginal difference} below \SI{0.1}{\%}. Sensitivity tests on a 70-day sample showed consistent and stable results within the range [\SI{0.001}{\%}, \SI{0.1}{\%}]. Looser thresholds occasionally prioritized false matches.
    \item \emph{Time gap} $\leq$ \SI{240}{\second}. This corresponds to the average LayerZero bridging time (see \Cref{tab:bridges_overview}), one of the most widely used cross-chain messaging protocols.
    \item Remaining candidates are ordered by ascending marginal difference, then by ascending time gap.
\end{enumerate}

For each transaction, we keep only its highest-ranked match, ensuring every pair is both economically plausible and temporally coherent.

\subsubsection{Identifying Bridge Usage}

To determine which strategy arbitrageurs are using—inventory or bridge arbitrage—we need a method to distinguish them. From a detection standpoint, our heuristics flag both types as cross-chain arbitrage. However, the key difference lies in the need for asset transfers: bridge arbitrage requires moving assets across chains, whereas inventory arbitrage does not. Therefore, by identifying whether bridge transactions occurred between the swaps, we can classify a cross-chain arbitrage.

Arbitrageurs typically use two types of bridges: native and multi-chain. \Cref{tab:bridges_overview} summarizes the bridges identified in this study, including their type, average bridging time, and detection method. Some, like LayerZero and Stargate, also serve as infrastructure for other bridges. We use two distinct approaches to detect and link bridge transactions within cross-chain arbitrages: \emph{unique identifier} and \emph{token transfer}.

\input{tables/bridges}

\textbf{Unique Identifier.} 
This method utilizes unique identifiers emitted by native bridges during cross-chain transactions, enabling us to link transfers even when sender and receiver addresses differ. Detection occurs in two phases: first, we compile a dataset of all native bridge transactions within the study period; then, we examine whether arbitrageurs used these native bridges between the legs of their arbitrage.

In the first phase, we adapt the method from \cite{rollinginshadows_torres} to link native bridge transactions across Arbitrum, Base, Optimism, Scroll, and ZKsync. We construct the dataset by scanning Ethereum and the corresponding rollups for events emitted by their native bridge contracts. For instance, on Arbitrum, we match Ethereum's \emph{InboxMessageDelivered} event with Arbitrum's \emph{RedeemScheduled} event, linking transactions via the ``message number''—a unique sequential ID generated by Arbitrum's bridge contract.

In the second phase, we examine whether a pair of native bridge transactions matches the address involved in the cross-chain arbitrage and falls within its time window. These transactions must also involve the same token used as the output of the first swap and the input of the second. If these conditions are met, we classify the arbitrage as bridge arbitrage; otherwise, as inventory arbitrage.

Note that we only consider native bridge transactions from Ethereum to rollups—not the other way around. This is because bridging to rollups typically takes a few minutes, while bridging back to Ethereum can take several days, making it impractical for cross-chain arbitrage due to the likely disappearance of price discrepancies over such a long period. Also note that the unique identifier method cannot be applied to Polygon, as its native bridge does not emit identifiable markers. For Polygon and multi-chain bridges, we instead rely on the token transfer-based method.

\textbf{Token Transfer.} This method links bridge transactions using ERC-20 token transfer events generated during the locking and minting process on either side of the bridge. While less precise than the unique identifier approach, it is more broadly applicable—especially for bridges that do not emit unique identifiers. We search for token transfers occurring between the two legs of an arbitrage and link them based on matching sender and receiver addresses involved in the detected cross-chain arbitrage.

We begin by scanning the source chain—where the first arbitrage leg occurs—block by block for transfer events involving the token last swapped in that leg. The search ends when we find a transfer where the sender matches the address that received the swapped tokens.
On the destination chain—where the second arbitrage leg takes place—we perform a reverse block-by-block search, starting from the second leg's block and moving backward. We look for transfer events involving the token first swapped in this leg, stopping when the recipient matches the address that initiated the second swap.

If matching token transfer events are found on both chains, as defined by the criteria above, we consider them part of the same bridge transaction. We then label the corresponding arbitrage as bridge arbitrage; if no such match is found, we classify it as inventory arbitrage.

\subsubsection{Profit Calculation}
For each matched transaction pair, we compute the \emph{net profit} in USD as
\begin{equation}
  \text{NetProfit}
  \;=\;
  \bigl(\text{USD}_{\mathrm{out,\,leg\,2}}
        -\text{USD}_{\mathrm{in,\,leg\,1}}\bigr)
  - \text{Costs},
\end{equation}
where \textbf{Costs} comprise
\begin{enumerate}
  \item gas fees for both swap transactions,
  \item any direct coinbase tips to the block builder\footnote{A coinbase transfer is an explicit ETH payment sent to the block’s coinbase address.},
  \item and, when present, bridge fees.  %
        (Only the source-chain call is user-initiated; the bridge operator typically submits the destination-chain leg.)
\end{enumerate}

\noindent
Gross revenue is simply the USD difference between the output asset of the second swap and the input asset of the first swap; subtracting the costs yields the net profit.

\subsection{Datasets}
We now describe the datasets we curated for our empirical analyses.

\subsubsection{Trading and Price Data}
We processed more than \SI{530}{million} trade transactions executed on nine blockchains over one year. Querying each chain via its own \gls{rpc} would be prohibitively costly, so we relied on Allium \cite{allium}, which provides transaction-level blockchain data. From Allium, we extracted swap and aggregator protocol events. Detecting swap events across multiple blockchains is inherently challenging due to the diversity of \glspl{dex} and aggregators, many of which use non-standard event signatures. Allium’s protocol tagging revealed \num{104} distinct protocols involved in the cross-chain matches, giving broad coverage. For every trade, we also collected metadata such as transaction fees and, when present, direct coinbase payments. USD volumes were computed with Allium’s hourly token price model \cite{allium_price}. We obtained per-second ETH and BTC candlestick data from Binance \cite{binance_market_data}, and used it to compute the assets' daily price volatility.

\subsubsection{Bridge Interactions}
Our token transfer-based bridge detection method leverages token transfer events extracted via Allium's raw logs dataset.
As our detection method is very generic and does not reveal the specific bridging protocol used, we rely on Allium’s labeled bridge dataset to identify the underlying bridge for each transaction—something not possible through token transfers alone. However, it is important to note that Allium's dataset covers only a subset of bridges and hence does not provide labels for all transactions.

\subsubsection{Smart Contract Labels}
To distinguish cross-chain arbitrage \gls{mev} bot contracts from unrelated contracts, we assembled a negative set of \SI{1285481}{} unique addresses. Sources include \gls{defi} protocol labels by Allium and Dune Analytics \cite{dune}, known non-\gls{mev} contract lists \cite{winnsterx}, and addresses of market makers we have identified.

\subsubsection{Ethereum Mempool Data}
We used Mempool Dumpster dataset by Flashbots \cite{mempool_dumpster}  to identify private Ethereum transactions. This dataset includes entries for Ethereum transactions observed by node providers in the mempool before being included in a block. Transactions missing from this dataset were likely privately relayed to Ethereum block builders through endpoints such as Flashbots Protect \cite{fbprotect} or MEV Blocker \cite{mevblocker}.

\subsection{Limitations}\label{limitations}
Our cross-chain arbitrage detection methodology identifies only two-hop cycles and relies on strict heuristics (marginal difference, temporal window, entity match), so results are a lower bound. Recently introduced \gls{defi} or bridge protocols not tracked by Allium are missing, and Allium's pricing model may lack data for less popular tokens, preventing USD volume estimates for those pairs. Inventory risk is not priced in arbitrage profits. Finally, Ethereum mempool data is limited by Mempool Dumpster's network coverage.

\subsection{Results}
We analyze the cross‑chain arbitrages detected by our methodology.

\subsubsection{Cross-Chain Arbitrage Landscape}\label{overview}

From September 2023 to August 2024, we identify \SI{242535}{} cross-chain arbitrages executed across nine blockchains. Collectively, they account for approximately \SI{868.64}{million} USD in trading volume, generating \SI{10.05}{million} USD in revenue, and yield \SI{8.65}{million} USD as net profit for arbitrageurs. We find that most activity (\SI{58.35}{\%}) occurs between an \gls{l1} and \gls{l2}, potentially due to the additional availability of native bridges. In contrast, \gls{l2}-\gls{l2} arbitrages (\SI{35.67}{\%}) must rely on third-party multi-chain bridges or pre-positioned inventory.

By count, Arbitrum hosts the largest share of arbitrage transactions (\SI{20.35}{\%}). However, Ethereum dominates in volume with \SI{36.44}{\%} and appears in all four highest-volume arbitraged chain-pairs, which also deliver the highest profits (see \Cref{tab:all_chain_pairs_metrics} in \Cref{chain_pairs_all}). Ethereum's prominence can be attributed to the substantial liquidity available on its \glspl{dex} \cite{defillama_eth} and its central role as the settlement layer for finality of many rollups that offer native bridges.

We track daily cross‑chain‑arbitrage volume and each chain’s share to uncover trends. \Cref{fig:daily_vol_and_vol_percentage_per_chain} shows that arbitrage activity gains traction over time. We find that average daily volume grows by 5.5$\times$, with a noticeable surge for Ethereum and several rollups beginning in March 2024. This jump aligns with Ethereum’s Dencun upgrade (March 13, 2024) \cite{dencun}, which introduced blob‑carrying transactions and cut rollup transaction fees—dropping, for example, from \SI{0.5}{USD} to \SI{0.003}{USD} on Base \cite{noauthor_base_nodate}—and thus making rollup-based cross-chain arbitrage more attractive.

\begin{figure}[b!]
    \centering
    \includegraphics[width=\linewidth]{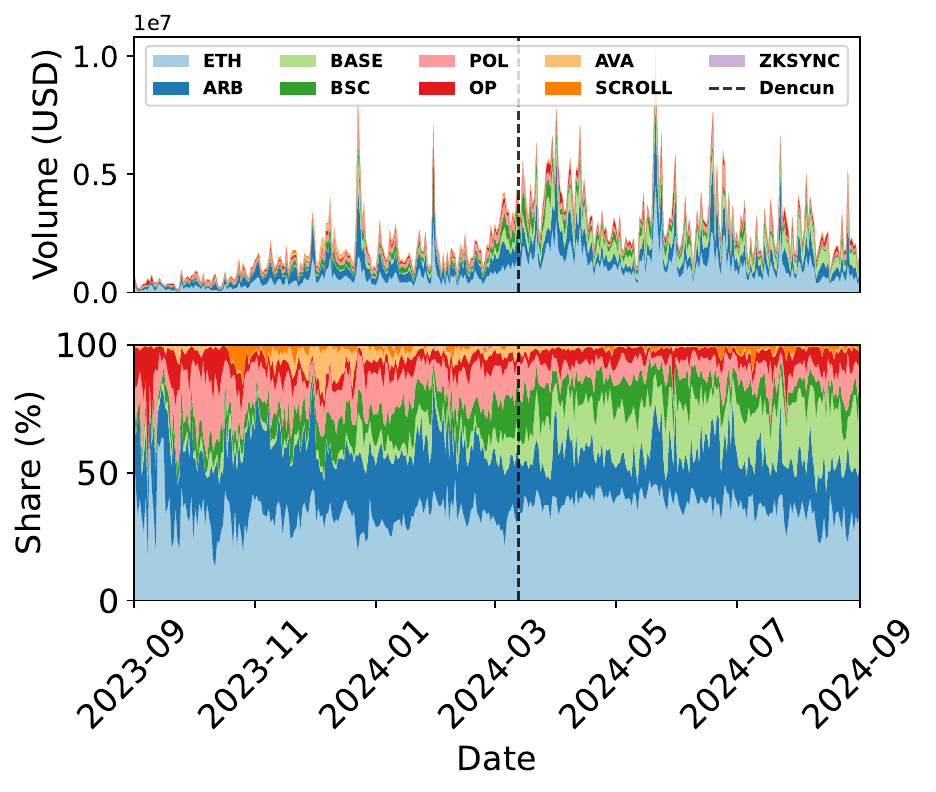}
    \caption{Daily USD cross‑chain arbitrage volume (top) and per‑chain share (bottom). The black dashed line marks Ethereum’s Dencun upgrade on March 13, 2024.}
    \label{fig:daily_vol_and_vol_percentage_per_chain}
\end{figure}

We assess Dencun’s impact on three daily aggregates—mean fee, trade count, and USD volume—using two-sided Welch $t$-tests (full results in \Cref{tab:activity_dencun_change}). Average per‑arbitrage cost falls from \SI{6.27}{USD} to \SI{4.61}{USD} (\(t_{365}=5.43,\;p=1.06\times10^{-7}\)); daily trade count rises from \SI{548}{} to \SI{788}{} (\(t_{350}=-8.63,\;p=2.19\times10^{-16}\)); and volume more than doubles from \SI{1.56}{million} to \SI{3.27}{million} (\(t_{314}=-11.9,\;p=2.93\times10^{-27}\)). These results show that the Dencun upgrade coincides with a clear reduction in per-trade costs and a large, statistically significant expansion in cross-chain-arbitrage activity.

\Cref{fig:daily_vol_and_vol_percentage_per_chain} also reveals rapid growth on Base. 
Besides the fee reductions via Dencun—which affect Base slightly more because it posts raw data, whereas, e.g., Arbitrum posts compressed calldata \cite{noauthor_understanding_2024}—liquidity on its \glspl{dex} also increases. Aerodrome’s \cite{aerodrome_dex} monthly volume on Base leaps from \SI{229}{million} USD in February 2024 to \SI{1.39}{billion} USD in March \cite{defillama_aero}, likely drawing arbitrageurs.

\input{tables/activity_after_dencun}

We observe several short-lived spikes in arbitrage volume, typically originating by a handful of token pairs, as illustrated in \Cref{fig:daily_token_volumes}. OMNI—the first token built on LayerZero's bridging technology \cite{omnicat}—drives activity following its launch on December 22, 2023 but collapses once a \gls{cex}, MEXC \cite{mexc-docs}, lists it on December 28. On the other hand, activity on BTC-ETH pairs
surges mid-May 2024, after Aerodrome lists tBTC—a bridged Bitcoin token introduced by the Threshold Network \cite{tbtc-docs} which alone accounts for \SI{69.19}{\%} of all BTC-ETH arbitrages. These cases highlight that better interoperability (e.g., via bridgeable tokens) can incentivize traders to engage in cross-chain arbitrage, even for major asset pairs, provided profitable market conditions exist. We present further details on arbitraged token pairs in \Cref{tokens}.

Beyond these token-specific bursts, daily cross-chain-arbitrage volume appears to track broader market volatility. Using ETH and BTC as proxies \cite{heimbach_non-atomic_2024}, we compute each day’s log high-to-low price ratio and plot it against the daily cross-chain arbitrage volume. \Cref{fig:daily_volume_and_volatility} reveals moderate yet highly significant Pearson correlations with arbitrage volume (ETH: $r=0.42,\;p=5.18\times10^{-17}$; BTC: $r=0.32,\;p=4.68\times10^{-10}$). This pattern suggests that price swings in major assets possibly propagate to other tokens through second-order effects, creating cross-chain arbitrage opportunities. 

\begin{figure}[b!]
    \centering
    \includegraphics[width=\linewidth]{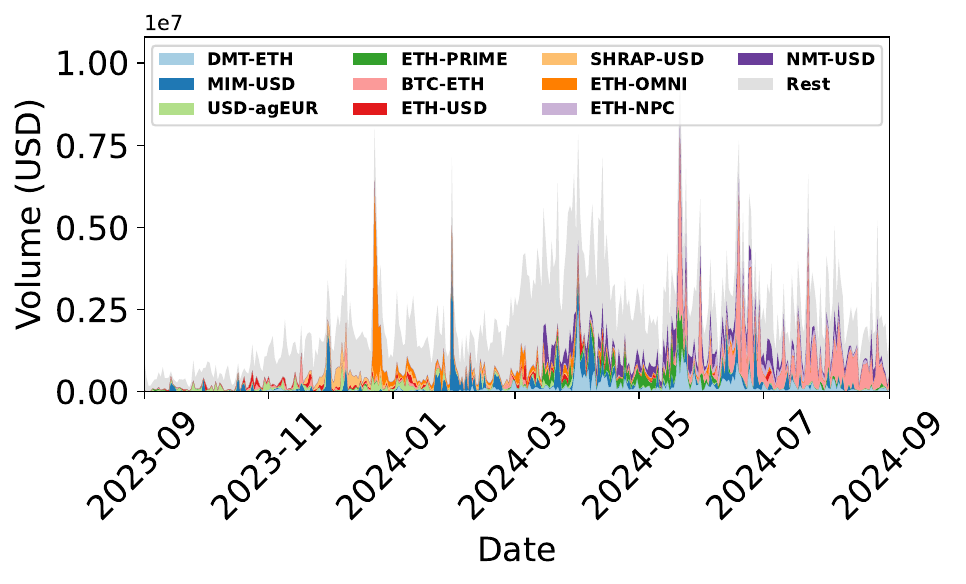}
    \caption{Daily USD volume for the ten token pairs with the highest cumulative arbitrage volume; all other pairs are grouped as ``Rest.''}
    \label{fig:daily_token_volumes}
\end{figure}

\subsubsection{Execution Methods: Inventory and Bridges}
We find that \SI{66.96}{\%} of cross-chain arbitrages rely on pre-positioned inventory across chains rather than transferring traded assets through a bridge. Among bridge-based trades, \SI{71.91}{\%} use multi-chain bridges while the remainder adopt native \gls{l1}-\gls{l2} bridges. \Cref{tab:chain_exec_types} breaks these figures down by chain‑pair category (i.e., \gls{l1}-\gls{l1}, \gls{l1}-\gls{l2}, \gls{l2}-\gls{l2}) and execution method, showing that in every category inventory arbitrage dominates\footnote{Our dataset labels 42 arbitrages between two \glspl{l2} as ``native‑bridge'' as each trade routes assets through its native bridge to Ethereum before reaching the destination \gls{l2}.}. 

\input{tables/arb_stats}

\begin{figure}[t!]
    \centering
    \includegraphics[width=\linewidth]{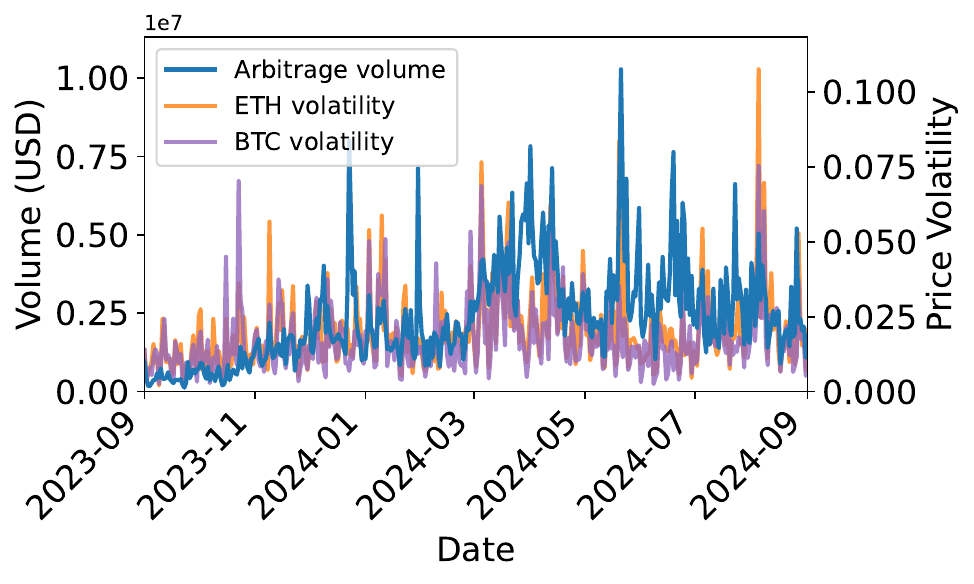}
    \caption{Daily USD cross-chain arbitrage volume against ETH and BTC price volatility.}
    \label{fig:daily_volume_and_volatility}
\end{figure}

A plausible driver of this inventory bias is the time required to complete both legs of the trade. To test that hypothesis, we measure the settlement time—the interval between the two legs of an arbitrage—and plot its distributions in \Cref{fig:inv_vs_bridge_duration}. Inventory trades settle fastest, with a \SI{9}{\second} median, followed by multi-chain‑bridge trades at \SI{162}{\second}, and native‑bridge trades lagging far behind at \SI{1030}{\second}, yielding an overall \SI{242}{\second} for bridge arbitrage. In other words, bridges impose a considerable \emph{latency cost}. A breakdown confirms that this cost is driven almost entirely by the asset-transfer leg: the interval between the outbound and inbound bridge transactions accounts for \SI{88.23}{\%} of total settlement time in native-bridge arbitrages and \SI{71.16}{\%} in multi-chain-bridge arbitrages. These delays likely explain why traders prefer the recurring cost of holding inventory over the opportunity cost of slow bridging.

We notice that settlement times for native‑bridge arbitrages show a bimodal pattern, with a secondary peak near \SI{1200}{\second}. More than half of these trades (\SI{11314}{}, \SI{50.27}{\%}) occur between Ethereum and Polygon—the most frequently arbitraged chain-pair (see \Cref{tab:all_chain_pairs_metrics} in \Cref{chain_pairs_all})—and settle in a median \SI{1248}{\second}. This latency possibly stems from Polygon’s inherently lengthy checkpoint‑based bridge design for securing chain's consistency and integrity \cite{polygonbridge-docs}. By contrast, only \SI{687}{} multi-chain‑bridge trades (\SI{1.19}{\%}) arbitrage the same pair, despite settling faster (\SI{415}{\second}). We discover that token availability appears to be driving this choice. Across the four most-arbitraged \gls{l1}-\gls{l2} pairs that offer native bridges\footnote{Ethereum-Scroll and Ethereum-ZKsync are omitted because they contribute negligible arbitrage activity.}, the asymmetry in bridge use correlates almost perfectly with the share of non-overlapping token pairs (Pearson \(r=0.98,\;p=0.02\)). Ethereum-Polygon sits at the extreme: \SI{94.27}{\%} of bridge-trades use the native one and \SI{97.81}{\%} of token pairs are exclusive to either the native or a multi-chain bridge, but not both. An arbitrageur without inventory must therefore pick whichever bridge supports the desired token—accepting potentially slower settlement (native) or counter-party risk (multi-chain)—constraints that inventory-based strategies avoid.

\begin{figure}[t!]
    \centering
    \includegraphics[width=\linewidth]{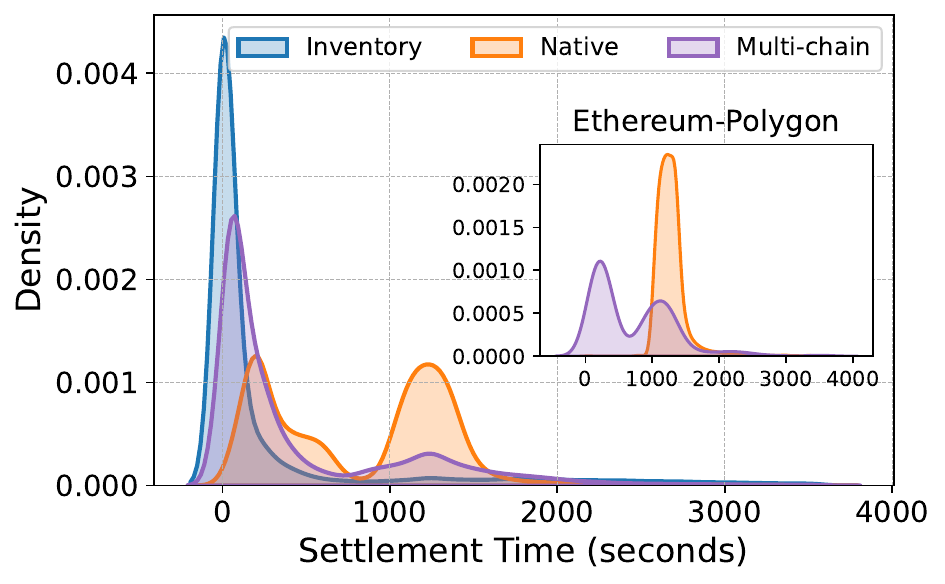}
    \caption{Distribution of settlement times for inventory, native-bridge, and multi-chain-bridge arbitrage. The inset highlights Ethereum-Polygon native- and multi-chain-bridge cases.}
    \label{fig:inv_vs_bridge_duration}  
\end{figure}

To examine competition, we look into how often inventory and bridge arbitrages are sent through private mempools—channels that hide transaction details from public (e.g., against frontrunning attacks) and offer features such as revert protection and bundling \cite{fbprotect, mevblocker}. However, measuring direct competition is difficult as most rollups lack mempool visibility. We therefore use Ethereum as a proxy since it hosts the highest-volume of activity and offers a public mempool. We observe \SI{35966}{} bridge arbitrages, of which \SI{42.81}{\%} adopt private mempools for the swap transaction—just over half of those (\SI{52.17}{\%}) also relay the bridge transaction privately—compared to only \SI{26.23}{\%} of \SI{52298}{} inventory arbitrages. The heavier reliance on private channels for bridge trades can imply higher competition, which is expected as capital requirements are lower, hence, there is a reduced barrier to entry. This can prompt traders to conceal their strategies and use bundles to ensure atomic execution for the swap and the bridge transaction so that either both succeed or both revert, mitigating partial execution risk.

We present summary statistics for bridge protocols appearing in our study in \Cref{bridge_protocols}.

\subsubsection{Who Controls Cross‑Chain Arbitrage?}
We observe a long tail of entities—\SI{9064}{} distinct addresses—engaged in cross‑chain arbitrage. Yet activity is highly skewed: \Cref{fig:searcher_cdf} shows that the top five addresses execute just over half of all trades, and the top four together command half of the total volume. Most strikingly, one address (\texttt{0xCA74}) alone accounts for more than a third of all volume.

To track concentration dynamics, we plot daily arbitrage activity for the largest traders with at least \SI{1}{\%} of total volume (see \Cref{tab:arber_metrics} in \Cref{arbers} for detailed statistics). \Cref{fig:arber_count_and_volume} reveals that \texttt{0xCA74} becomes progressively more dominant, likely benefiting from the Dencun upgrade: the arbitrageur's share climbs from \SI{20.3}{\%} before Dencun to \SI{39.7}{\%} afterwards (\(t_{340}=-18.6,\;p=8.14\times10^{-54}\)). Over the same window it more than doubles its daily trade count (\SI{79.4}{} to \SI{175}{}, \(t_{341}=-17.3,\;p=9.44\times10^{-49}\)) and increases its daily volume five‑fold (\SI{0.35}{million} to \SI{1.3}{million}; \(t_{238}=-14,\;p=6.52\times10^{-33}\)), while its mean fee per trade falls from \SI{10.0}{USD} to \SI{6.57}{USD} (\(t_{364}=5.42,\;p=1.08\times10^{-7}\)). Full statistics for \texttt{0xCA74}'s activity before and after Dencun are provided in \Cref{tab:top_searcher_dencun_change} in \Cref{arbers}.

We notice that other large actors demonstrate intermittent patterns. For instance, \texttt{0x6226}—second by trade count—avoids trades involving Ethereum, thereby limiting its volume. By contrast, \SI{74.46}{\%} of \texttt{0xCA74}’s arbitrages include Ethereum, suggesting that meaningful market dominance demands the capacity to commit large capital to high‑volume pairs. This capital barrier sets cross‑chain arbitrage apart from single‑chain atomic arbitrage, where flash loans eliminate the need for upfront funds.

We find that 
only three out of the eleven leading arbitrageurs submit trades through \gls{mev} smart contracts, and the rest act directly via \glspl{eoa}. This pattern likely stems from the non-atomic nature of cross-chain arbitrage: a trader cannot embed on-chain logic that both detects an opportunity across multiple chains and executes only if it remains profitable, as is common in single-chain atomic arbitrage. Since no protocol yet guarantees cross-chain atomic execution, traders must build off-chain infrastructure to monitor markets and deploy transactions at the right moment—bearing latency and execution risk. This engineering burden may form a competitive moat and further concentrate the market.

\begin{figure}[t!]
  \centering
  \includegraphics[width=\linewidth]{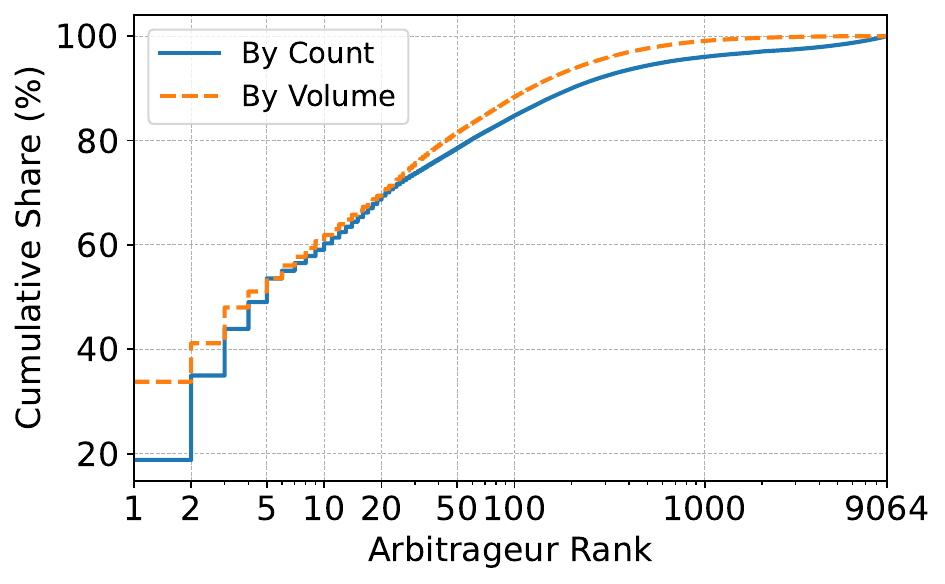}
  \caption{Cumulative distribution of cross‑chain arbitrage count (solid blue) and volume (dashed orange) by arbitrageur.}
  \label{fig:searcher_cdf}
\end{figure}

\begin{figure}[t!]
  \centering
  \includegraphics[width=\linewidth]{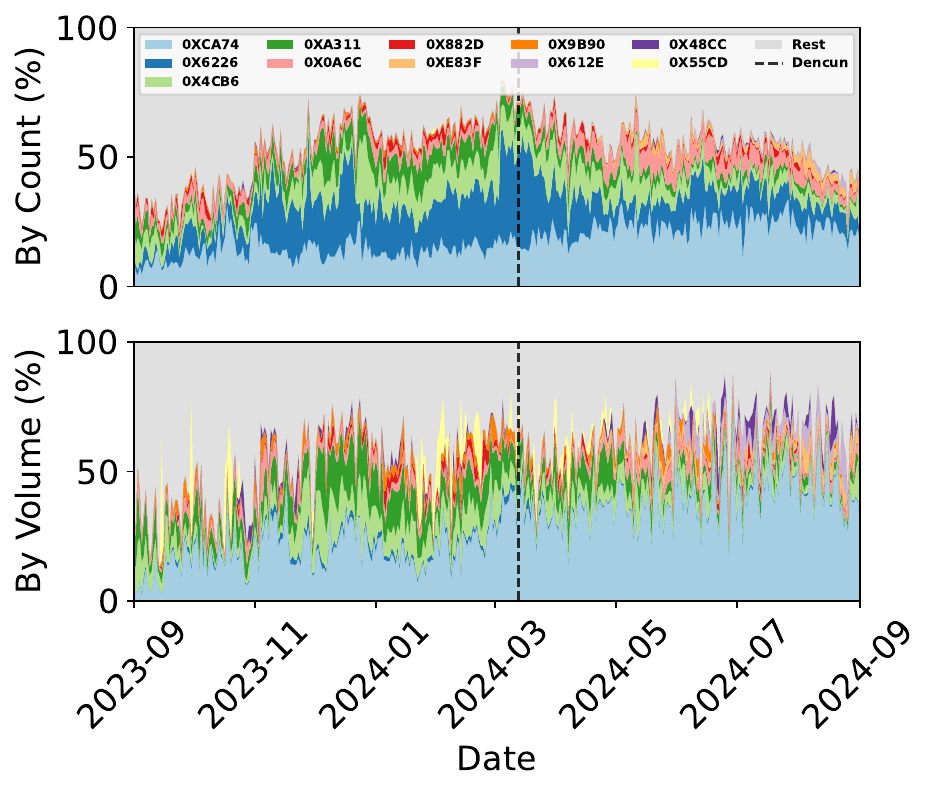}
\caption{Daily share of cross-chain-arbitrage count (top) and volume (bottom) for addresses with at least 1 \% of total volume; all remaining addresses are grouped as ``Rest.'' The black dashed line marks Ethereum’s Dencun upgrade on March 13, 2024.}
  \label{fig:arber_count_and_volume}
\end{figure}
\section{Discussion}
Cross-chain arbitrage 
is mainly lucrative to actors who can minimize execution risks—especially latency—and manage inventory across chains. Sequencing control addresses the first challenge: an entity that determines block order on several chains can ensure its multi-chain transaction bundle lands exactly as intended. The second challenge demands trading skill and capital to deploy the right size of inventory on the right chain at the right time.

Together, these requirements create a powerful incentive for \emph{vertical integration}. \gls{hft} shops and sophisticated \gls{mev} searchers (e.g., \texttt{0xCA74}) can pair their technical know-how with the sequencing leverage of infrastructure operators like centralized rollup sequencers or large staking providers running validators across many \gls{pos} chains. A similar dynamic is already visible on Ethereum, where several \gls{hft} firms now run block builders to execute their strategies with minimal delay \cite{yang2024decentralizationethereumsbuildermarket,oz_et_al:LIPIcs.AFT.2024.22,bahrani_centralization_2conf024,gupta_centralizing_2023, heimbach_non-atomic_2024}. By merging searcher logic with sequencing control, these alliances concentrate both physical infrastructure (clustered in few data-centers) and economic power, creating a feedback loop: \emph{sequencing advantages yield higher \gls{mev}, which funds still more sequencing reach}. The result is a small set of centralized gatekeepers, which introduces three systemic hazards: \emph{censorship risk}—economically motivated suppression of competing \gls{mev} bundles or compliance with regulatory sanctions \cite{wahrstatter_blockchain_2023}; \emph{liveness risk}—a centralized sequencer's outage or block withholding can stall multiple chains simultaneously; and \emph{finality risk}—multi-chain ``time-bandit'' re-orgs that rewrite history to extract additional \gls{mev} from past blocks \cite{obadia_unity_2021}.

\smallskip
\noindent\textbf{Decentralized block building.}  
To mitigate the risks that follow from \gls{mev}-driven economic and infrastructural concentration, block building must be decentralized. On \glspl{l1} like Ethereum—with already large validator sets—several in-protocol proposals tackle the problem \cite{thiery_fork-choice_2024, cryptoeprint:2025/194, fox_et_al:LIPIcs.AFT.2023.19, neuder_concurrent_2024}. They 
seek to break a single block proposer's monopoly by (\textbf{i}) obliging the proposer to include transactions supplied by other validators (e.g., through inclusion lists \cite{noauthor_eip-7547_nodate}), or (\textbf{ii}) assigning multiple proposers to the same slot and merging their blocks deterministically (e.g., by priority-fee order). Outside the protocol, BuilderNet \cite{buildernet} offers a complementary approach: a decentralized network of \gls{tee}-backed builders that share the same order flow, so a transaction censored by one builder can still be included by another \cite{buildernet}. Analogous solutions apply to rollups, which now rely on a single sequencer.
One path is the based rollup model, which delegates block building to Ethereum validators \cite{noauthor_based_2023}; this inherits Ethereum’s security and censorship resistance, but also its throughput bottleneck. An alternative is a shared sequencer network such as Espresso \cite{noauthor_espresso_doc_nodate}, which replaces the single rollup sequencer with a decentralized ordering layer: the network only orders transactions—leaving execution to the rollup—thereby reducing unilateral control while maintaining high throughput.

\smallskip
\noindent\textbf{Lowering entry barriers.}  
A complementary way to counteract concentration is to broaden participation in cross-chain arbitrage. When \gls{mev} extraction becomes less technically demanding, competition increases and profit margins narrow—as seen in single-chain atomic arbitrage \cite{taleoftwoarbs}. Reduced margins can weaken the incentive for any infrastructure operator to form an exclusive alliance with an individual \gls{mev} searcher. 

Lowering barriers requires cutting both inventory overhead and execution latency. \emph{Atomic, instant bridging} across \glspl{l1} and \glspl{l2} would let traders move assets inside the same bundle, so they can shift funds just-in-time instead of pre-positioning and continuously rebalancing them. This removes bridge delay and unlocks cross-chain flash loans—a trader can borrow on Chain A, execute leg 1, atomically bridge, execute leg 2 on Chain B, and repay, all in one bundle. A non-atomic variant already exists (borrow on Chain A against collateral, trade on both chains, then bridge back to repay and reclaim the collateral), as illustrated in \Cref{fig:loan_arb} and detailed in \Cref{loan-arb}. By reducing capital and operational hurdles, atomic bridging widens participation and helps neutralize the centralizing forces discussed above.

\section{Related Work}
Qin et al. \cite{qin2022quantifying} provided one of the first comprehensive quantification of \gls{mev}, documenting systematic value extraction through sandwich attacks, liquidations, and \gls{dex} arbitrage on Ethereum. 
Building up on the theoretical work presented by Eskandari et al. \cite{EskandariMC19}, Torres et al. \cite{frontrunner_jones_torres_2021} presented the first empirical quantification of frontrunning strategies.
McLaughlin et al. \cite{mclaughlin} advanced this understanding by developing an application-agnostic arbitrage detection methodology using standardized ERC-20 transfer events, enabling analysis across a broader range of decentralized exchanges. Torres et al. \cite{rollinginshadows_torres} extended \gls{mev} research into \gls{l2} solutions by investigating \gls{mev} activities across Arbitrum, Optimism, and ZKsync, revealing comparable trading volumes to Ethereum on rollups, albeit with significantly lower profits despite reduced costs. Öz et al. \cite{fcfsmevalgorand} provided empirical evidence of \gls{mev} extraction on \gls{fcfs} blockchains. The analysis of non-atomic arbitrage has been advanced by Heimbach et al. \cite{heimbach_non-atomic_2024}, who demonstrated that around \SI{25}{\%} of volume on Ethereum's top five \glspl{dex} could be attributed to arbitrage between \glspl{cex} and \glspl{dex}, revealing substantial centralization among a small group of dominant \gls{mev} searchers. Obadia et al. \cite{obadia_unity_2021} provided the first formal definition of cross-domain \gls{mev} by introducing domains as self-contained systems with shared states and showing how cross-domain \gls{mev} opportunities incentivize sequencer collusion across chains. McMenamin \cite{mcmenamin2023sokcrossdomainmev} systematized cross-domain \gls{mev} by categorizing extractable value based on extraction methods and value origins while analyzing different protocols' \gls{mev} mitigation capabilities. Mazor et al. \cite{mazor_empirical_2023} examined hypothetical cross-chain arbitrage opportunities between PancakeSwap on Binance Smart Chain and QuickSwap on Polygon, providing insights into the magnitude of extractable arbitrage value and the typical duration of these opportunities. Gogol et al. \cite{gogol_cross-rollup_2024} analyzed price disparities in WETH-USDC trading pairs between Ethereum and major \gls{l2} rollups, showing that arbitrage opportunities persisted for multiple blocks with relatively limited profits ranging from \SI{0.03}{\%} to \SI{0.25}{\%} of trading volume.  Mazorra et al. \cite{Mazorra2022PriceOM} provided theoretical foundations by developing an abstract framework for analyzing \gls{mev} games across domains, particularly focusing on network characteristics' influence on strategies and congestion. Finally, Öz et al.'s \cite{oz_et_al:LIPIcs.AFT.2024.22} and Yang et al.'s \cite{yang2024decentralizationethereumsbuildermarket} analyses on block building revealed how private order flow access by select builders leads to market centralization in Ethereum's block building, highlighting similar centralization risks that could emerge from privileged access to cross-chain \gls{mev} opportunities.

\begin{figure}[t!]
    \centering
    \includegraphics[width=\linewidth]{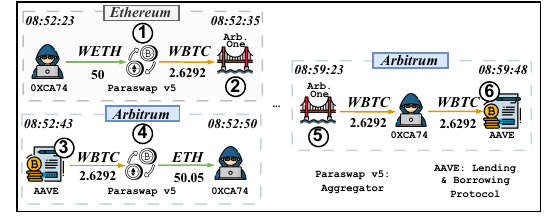}
    \caption{A loan-backed cross-chain arbitrage between Ethereum and Arbitrum. Full transaction details in \Cref{loan-arb}.}
    \label{fig:loan_arb}  
\end{figure}

\section{Conclusion}
As trading migrates fully on-chain across multiple networks, cross-chain arbitrage is poised to become \gls{defi}'s primary mechanism for keeping prices aligned across chains. We examine this mechanism with a profit–cost model contrasting inventory- and bridge-based execution, and with a year-long measurement across nine blockchains. The model pinpoints opportunity frequency, bridge time, and token depreciation as the key profitability drivers; empirically, daily volume grew 5.5$\times$ over the study window and spiked after Ethereum's Dencun upgrade. Yet bridges still impose a sizable latency cost: \SI{66.96}{\%} of arbitrages use inventory and settle in \SI{9}{\s}, whereas bridge trades take \SI{242}{\second}. Activity is concentrated as well—five addresses generate over half of all trades, and one captures nearly \SI{40}{\%} of daily volume post-Dencun. Because this \gls{mev} stream rewards vertical integration between searchers and infrastructure operators such as sequencers or staking providers, it risks centralizing sequencing infrastructure and economic power, amplifying censorship, liveness, and finality threats. Decentralizing block building and lowering entry barriers—thereby undercutting exclusive sequencer–searcher deals—are critical countermeasures. We hope these findings inform designs that curb such risks as the ecosystem advances into an ever more multi-chain future.

\section*{Acknowledgements}
This work is supported by the European Union’s Horizon 2020 research and innovation programme under grant agreement No 952226, project BIG (Enhancing the research and innovation potential of Técnico through Blockchain technologies and design Innovation for social Good).
We extend our gratitude to Danut Ilisei, Murad Muradli, and Thomas Wagner for their contributions. We also thank Quintus Kilbourn and Sarah Allen for their valuable feedback.

\bibliographystyle{ACM-Reference-Format}
\bibliography{sample-base}

\appendix
\section{Cost of Inventory with Bounded Liquidity}\label{coi_bounded}
The formulae for the cost of liquidity calculation can be adapted to the case of bounded liquidity in the second market, as we sketch next:

Suppose the second market has bounded liquidity and is also operated as a CPMM with reserves $\tilde{R}_t^A$ and $\tilde{R}_t^B$. A good approximation of a CPMM is that of quadratic trading cost. More specifically, define the trading cost of a CPMM to be the difference in exchanged tokens between a perfectly liquid market and the CPMM,
$$\tilde{C}_t(x):=Q_t x-\frac{\tilde{R}_t^Ax}{\tilde{R}_t^B+x}=\frac{Q_t x(\tilde{R}_t^B+x)-\tilde{R}_t^Ax}{\tilde{R}_t^B+x}=\frac{Q_t x^2}{\tilde{R}_t^B+x}.$$
Then for small trade sizes relative to the total liquidity in the pool, we can approximate trading cost by
$$\tilde{C}_t(x)\approx \frac{Q_t x^2}{\tilde{R}_t^B}.$$
The cost of inventory has now a cost term:
\begin{align*}C(I)&=E_0[\int_0^\tau \frac{Q_t}{\tilde{R}_t^B}(dI_t)^2-\int_0^\tau I_tdQ_t]\\&=E_0[\int_0^\tau \frac{Q_t}{\tilde{R}_t^B}\sigma^2\left(\frac{\partial I_t}{\partial Q_t}\right)^2dt-\int_0^\tau I_tdQ_t].\end{align*}
For the case of reserves $$R_t^B=\frac{Q_0^k}{Q_t^k}R_0^B$$
and a constant multiply of reserves in the second AMM, $\tilde{R}_t^A=\phi R_A$ and $\tilde{R}_t^B=\phi R_B$, we obtain for example:
\begin{align*}\frac{C(I)}{E_0[Q_{\tau} I_{\tau}]}&=\frac{\lambda+(1-k)(\tfrac{1}{2}k\sigma^2-\mu)}{\lambda R_0^A}E_0[\int_0^\tau \frac{Q_t}{\phi R_t^B}\left(\frac{kR_t^B}{Q_t}\right )^2\sigma^2dt]--\frac{\mu}{\lambda}\\&=\frac{\lambda+(1-k)(\tfrac{1}{2}k\sigma^2-\mu)}{\lambda R_0^A}E_0[\int_0^\tau \frac{k^2}{\phi Q_t}R_t^B\sigma^2dt]-\frac{\mu}{\lambda}\\&=\frac{1-k}{1+k}\frac{(\lambda/(1-k)-\mu+\tfrac{1}{2}k\sigma^2)k^2\sigma^2}{\phi\lambda Q_0^2(\lambda/(k+1)+\mu-\tfrac{1}{2}\sigma^2(k+2))}-\frac{\mu}{\lambda}.\end{align*}
In the case of $k=0$, we obtain:
$$-\frac{\mu}{\lambda}$$
and in the case of $k=1/2$ we obtain:
$$\frac{1}{3}\frac{(2\lambda-\mu+\tfrac{1}{4}\sigma^2)\tfrac{1}{4}\sigma^2}{\phi\lambda Q_0^2(\tfrac{3}{2}\lambda+\mu-5\sigma^2))}-\frac{\mu}{\lambda}$$
\section{Cross-Chain Arbitrage Examples}\label{examples}
We now breakdown real-world examples of an inventory arbitrage and a bridge arbitrage.

\subsection{Case Study: Inventory Arbitrage}\label{inv-arb}
\Cref{tab:sia} summarizes a real-world example of an inventory arbitrage—see \Cref{fig:genaral-inv-arb} for an illustration. The arbitrageur first swaps \SI{0.9}{WETH} for \SI{1034616.49}{OMNI} on Base at 12:25:51 AM. On Arbitrum, they swap \SI{1034616.49}{OMNI} for close to \SI{1}{WETH}, earning approximately \SI{198.95}{USD} in profit after deducting transaction fees. The entire arbitrage gets completed in \SI{27}{\second}. 
Post-arbitrage, the arbitrageur rebalances their inventory by transferring \SI{1034616.49}{OMNI} from Base to Arbitrum using the LayerZero bridge\footnote{Rebalance transactions: Base \href{https://basescan.org/tx/0x11e0c5acc495159aa935b25a8e868bfe2bcf7cf7833df264168a0e7f745d3dcf}{\text{0x11e...dcf}}, Arbitrum \href{https://arbiscan.io/tx/0x2b855cfddb8059cf6332f07cd8c58637bf19a5dfa7a7d28c3de8a1877bf57dca}{\text{0x2b8...7dca}}.}.

\input{tables/sia_example}

\subsection{Case Study: Bridge Arbitrage}\label{bridge-arb}
\Cref{tab:sda} summarizes a real-world example of a bridge arbitrage—see \Cref{fig:genaral-bridge-arb} for an illustration. The arbitrageur initially swaps \SI{20000}{USDT} for \SI{2692596.48}{VOW} on Ethereum at 11:13:11 AM. Following this, they approve DBridge to transfer VOW tokens and deposit \SI{2692590}{VOW} \SI{60}{\second} later. 
The arbitrageur receives \SI{2692590}{VOW} on Binance Smart Chain at 11:19:10 AM and, \SI{12}{\second} later, swaps \SI{2692593.08}{VOW} for \SI{134889.18}{} Binance-Peg BSC-USD, earning around \SI{114688.38}{USD} in profit after transaction fees. The entire arbitrage takes over \SI{6}{minutes} to complete.

\input{tables/sda_example}

\section{Blockchain Pairs}\label{chain_pairs_all}

\Cref{tab:all_chain_pairs_metrics} summarizes the cross-chain arbitrage metrics for all blockchain pairs in our study.

\section{Tokens}\label{tokens}

\Cref{tab:top_token_pairs_stats} reports summary statistics for the ten token pairs with the largest cumulative arbitrage volume; together they account for \SI{43.05}{\%} of total volume. The remaining \SI{1858}{} pairs share the other \SI{56.95}{\%}.

\section{Bridges}\label{bridge_protocols}
\Cref{tab:bridge_protocol_stats} summarizes the usage metrics of native and multi-chain bridges for cross-chain arbitrages, excluding those used in less than \SI{0.1}{\%} of all bridge arbitrage. The \textit{Cost Ratio} column shows, for each bridge, the share of total transaction fees attributable to the bridge leg relative to the aggregate fees of the entire arbitrage.

\section{Arbitrageurs}\label{arbers}
\paragraph{0xCA74 Activity.} \Cref{tab:top_searcher_dencun_change} reports the statistics for the cross-chain arbitrage activity of the arbitrageur \texttt{0xCA74}, before and after the Dencun upgrade.
\paragraph{Summary Statistics}
\Cref{tab:arber_metrics} summarizes the cross-chain arbitrage metrics for arbitrageurs executing $\geq$\SI{1}{\%} of the total trade volume; together they account for \SI{63.01}{\%} of total volume. The remaining \SI{9053}{} traders conduct the other \SI{37.99}{\%}.

\section{Case Study: Loan-Backed Arbitrage}\label{loan-arb}
When an arbitrageur lacks pre-positioned inventory, the usual fallback is to bridge assets across chains—incurring delay and execution risk.  
A faster workaround is to \emph{borrow} the required tokens from a lending protocol, execute the arbitrage instantaneously as though the inventory were already in place, and then bridge the proceeds back to repay the loan.  
This loan-backed method replicates inventory arbitrage without the upfront capital commitment, at the cost of paying loan fees and taking temporary price risk during the bridge leg.

\Cref{fig:loan_arb} shows a real-world loan-backed arbitrage executed by \texttt{0xCA74} between Ethereum and Arbitrum.
The process unfolds as follows:
\begin{enumerate}
\itemsep0em 
    \item \texttt{0xCA74} initially swaps \SI{50}{WETH} for \SI{2.629}{WBTC} on Ethereum (\href{https://etherscan.io/tx/0xf9dac896022ab70b8831e4a0139b97738180bf1605f216400f66bb0957f9208a}{0xf9d...08a}).
    \item The WBTC is bridged to Arbitrum in the next block (\href{https://etherscan.io/tx/0x452a9c05cc02fe146ce13a788456f02bf68c83ee3d6a57c0402ccfd8c5108903}{0x452...903}).
    \item \SI{8}{\second} later, \texttt{0xCA74} takes a loan for \SI{2.629}{WBTC} using Aave \cite{aave} (\href{https://arbiscan.io/tx/0xf7e5f28d391c2dee5f982fd27f8caf6fc250acd54ebe3d06ace25ac2a85feeff}{0xf7e...eff}).
    \item The WBTC is then swapped back to ETH, earning \SI{0.5}{ETH} and completing the arbitrage 
    (\href{https://arbiscan.io/tx/0xdd45995d12a9404a667c13ec0ab014b927abeea3cf9cd650a5e8e63845c7972f}{0xdd4...72f}).
    \item The bridged WBTC arrives on Arbitrum around \SI{7}{minutes} after the arbitrage is completed (\href{https://arbiscan.io/tx/0x7015f3ceff63cc5e4c2e7abffe3ed6c1d21da6b3305724ed1c9177e50cf411bf}{0x701...1bf}).
    \item Finally, \texttt{0xCA74} transfers the bridged WBTC to AAVE to close the loan, concluding the strategy (\href{https://arbiscan.io/tx/0x798b5371e00f8a0d5dd7471aa36488e490f853004e638b547f272013b19db935}{0x798...935}).
\end{enumerate}

\input{tables/chain_pairs_stats_formatted}

\input{tables/top_token_pairs_stats}
\input{tables/bridge_protocol_stats}

\input{tables/top_searcher_dencun}

\input{tables/searcher_stats_formatted}

\end{document}

%% file: acronyms.tex
\newacronym{pbs}{PBS}{Proposer-Builder Separation}
\newacronym{ofc}{OFC}{Order Flow Composition}
\newacronym{bft}{BFT}{Byzantine Fault Tolerant}
\newacronym{pos}{PoS}{Proof-of-Stake}
\newacronym{pow}{PoW}{Proof-of-Work}
\newacronym{mev}{MEV}{Maximal Extractable Value}
\newacronym[shortplural=DEXes]{dex}{DEX}{Decentralized Exchange}
\newacronym[shortplural=CEXes]{cex}{CEX}{Centralized Exchange} 
\newacronym{amm}{AMM}{Automated Market Maker}
\newacronym{defi}{DeFi}{Decentralized Finance}
\newacronym{tradfi}{TradFi}{Traditional Finance}
\newacronym{pga}{PGA}{Priority Gas Auction}
\newacronym{tob}{ToB}{Top-of-Block}
\newacronym{bob}{BoB}{Body-of-Block}
\newacronym{eob}{EoB}{End-of-Block}
\newacronym{ofa}{OFA}{Order Flow Auction}
\newacronym{0xce91228789b57deb45e66ca10ff648385fe7093b}{0xCe91228789B57DEb45e66Ca10Ff648385fE7093b}{MEV Blocker Rebates Safe}
\newacronym{rfq}{RFQ}{Request For Quote}
\newacronym{pm}{PM}{Profit Margin}
\newacronym{eoa}{EOA}{Externally Owned Account}
\newacronym{eof}{EOF}{Exclusive Order Flow}
\newacronym{msof}{MSOF}{Most Significant Order Flow}
\newacronym{jit}{JIT}{Just-In-Time}
\newacronym{lda}{LDA}{Linear Discriminant Analysis}
\newacronym{da}{DA}{Decoding Accuracy}
\newacronym{rpc}{RPC}{Remote Procedure Call}
\newacronym{brt}{BRT}{\texttt{beaverbuild}, \texttt{rysnc}, and \texttt{Titan}}
\newacronym{ep}{EP}{Exclusive Provider}
\newacronym{tee}{TEE}{Trusted Execution Environment}
\newacronym{il}{IL}{Inclusion List}
\newacronym{aps}{APS}{Attestor-Proposer Separation}
\newacronym{ea}{EA}{Execution Auction}
\newacronym{cr}{CR}{Censorship-Resistance}
\newacronym{hft}{HFT}{High-Frequency Trading}
\newacronym{l2}{L2}{Layer-2}
\newacronym{l1}{L1}{Layer-1}
\newacronym{bsc}{BSC}{Binance Smart Chain}
\newacronym{zk}{ZK}{Zero Knowledge}
\newacronym{tvl}{TVL}{Total Value Locked}
\newacronym{sia}{SIA}{Sequence-Independent Arbitrage}
\newacronym{sda}{SDA}{Sequence-Dependent Arbitrage}
\newacronym{cdf}{CDF}{Cumulative Distribution Function}
\newacronym{ci}{CI}{Confidence Interval}
\newacronym{fp}{FP}{False Positive}
\newacronym{fn}{FN}{False Negative}
\newacronym{clob}{CLOB}{Central Limit Order Book}
\newacronym{p2p}{P2P}{Peer-to-Peer}
\newacronym{fdv}{FDV}{Fully Diluted Valuation}
\newacronym{cpmm}{CPMM}{Constant Product Market Maker}
\newacronym{fcfs}{FCFS}{First-Come, First-Served}
\newacronym{hhi}{HHI}{Herfindahl–Hirschman Index}
\newcommand{\encircled}[2][0.9mm]{%
    \raisebox{.8pt}{%
        \textcircled{%
            \raisebox{0.35pt}{%
                \kern #1
                \scalebox{0.70}{#2}
            }%
        }%
    }%
}
\makeatletter
\newcommand*{\ensquared}[1]{\relax\ifmmode\mathpalette\@ensquared@math{#1}\else\@ensquared{#1}\fi}
\newcommand*{\@ensquared@math}[2]{\@ensquared{$\m@th#1#2$}}
\newcommand*{\@ensquared}[1]{%
\tikz[baseline,anchor=base]{\node[draw,outer sep=0pt,inner sep=0.6mm,minimum width=3.8mm] {#1};}} 
\makeatother

%% file: tables/blockchains_overview.tex
\begin{table}[t!]
\centering
\footnotesize
\setlength{\tabcolsep}{2pt}
\begin{tabular}{llllrr}
\toprule
\textbf{Blockchain Name} & \textbf{Type} & \textbf{Mempool} & \textbf{Ordering} & \textbf{Block Time} &  \textbf{Launch}  \\
\midrule
Avalanche (AVA)       & L1         & Public  & Gas Price & $\sim$\SI{2.00}{\second}  & Sep 2020   \\
Binance Smart Chain (BSC)& L1     & Public  & Gas Price  & $\sim$\SI{3.00}{\second} & Sep 2020   \\
Ethereum (ETH)        & L1         & Public  & Gas Price       & $\sim$\SI{12.00}{\second}& Jul 2015  \\ 
Polygon PoS (POL)    & L2-SC      & Public  & Gas Price  & $\sim$\SI{2.00}{\second} & Jun 2020       \\ 
Arbitrum (ARB)       & L2-OR      & Private & FCFS$^{\ast}$     & $\sim$\SI{0.25}{\second} & Aug 2021  \\ 
Base (BASE)           & L2-OR      & Private & Gas Price$^{\ast\dagger}$     & $\sim$\SI{2.00}{\second} & Aug 2023 \\ 
Optimism (OP)       & L2-OR      & Private & Gas Price$^{\ast}$   & $\sim$\SI{2.00}{\second} & Dec 2021  \\  
Scroll (SCROLL)          & L2-ZK      & Private & FCFS$^{\ast}$   & $\sim$\SI{3.50}{\second} & Oct 2023   \\ 
ZKsync Era (ZKSYNC)     & L2-ZK      & Private & FCFS$^{\ast}$      & $\sim$\SI{1.00}{\second} & Mar 2023 \\  \bottomrule
$^{\ast}$Centralized sequencer.\\
$^{\dagger}$FCFS tie-breaking.
\end{tabular}
\caption{Overview of investigated Layer-1 (L1) and Layer-2 (L2) solutions. SC - Sidechain, OR - Optimistic Rollup, ZK - ZK-Rollup, FCFS - First-Come, First-Served.}
\label{tab:blockchains_overview}
\end{table}

%% file: tables/bridges.tex
\begin{table}[b!]
\centering
\footnotesize
\setlength{\tabcolsep}{2pt}
\begin{adjustbox}{width=\columnwidth}
\begin{tabular}{llrc}
\toprule
\textbf{Bridge Name} & \textbf{Bridge Type} & \textbf{Bridging Time} & \textbf{Detection Method} \\
\midrule
Across~\cite{acrossbridge-docs}      & Multi-chain & 	1-4 mins & Token Transfer \\
Arbitrum One~\cite{arbbridge-docs}     & Native & 15-30 mins & Unique Identifier \\
Axelar~\cite{axelarbridge-docs}      & Multi-chain &  	up to 1 hour & Token Transfer \\
Base~\cite{basebridge-docs}      & Native & $\sim$10 mins & Unique Identifier \\
Celer~\cite{celerbridge-docs}      & Multi-chain &  	5-20 mins & Token Transfer \\
LayerZero~\cite{layerzerobridge-docs}      & Multi-chain & $\sim$4 mins & Token Transfer \\
Optimism~\cite{opbridge-docs}      & Native & 1-3 mins & Unique Identifier \\
Polygon~\cite{polygonbridge-docs}      & Native & 10-30 mins & Token Transfer \\
Scroll~\cite{scrollbridge-docs}      & Native & $\sim$4 hours & Unique Identifier \\
Stargate~\cite{layerzerobridge-docs}      & Multi-chain & $\sim$4 mins & Token Transfer \\
Synapse~\cite{synapsebridge-docs}      & Multi-chain & 	10-20 mins & Token Transfer \\
Wormhole~\cite{wormholebridge-docs}      & Multi-chain &  	up 24 hours & Token Transfer\\
ZKsync~\cite{zksyncbridge-docs}      & Native & $\sim$15 mins & Unique Identifier \\
\bottomrule
\end{tabular}
\end{adjustbox}
\caption{Overview of detected native and multi-chain bridges.}
\label{tab:bridges_overview}
\vspace{-5mm}

\end{table}

%% file: tables/activity_after_dencun.tex

\begin{table}[t!]
\centering
\footnotesize     
\resizebox{\columnwidth}{!}{%
\begin{tabular}{lrrrrrrrr}
\toprule
\textbf{Metric} & \textbf{Pre} & \textbf{Post} & $\Delta$ & \textbf{95\% CI} & $t$ & $df$ & $p$ & $d$ \\
\midrule
Fee (USD)    & 6.27 & 4.61 & -1.66 & $[-2.26,-1.06]$              &  5.43 & 365 & 1.06e‑7  & -0.56 \\
Count        & 548  & 788  &  240  & $[186,295]$                  & -8.63 & 350 & 2.19e‑16 &  0.88 \\
Volume (USD) & 1.56e6 & 3.27e6 & 1.7e6 & $[1.42,1.99]\!\times\!10^{6}$ & -11.91 & 314 & 2.93e‑27 &  1.27 \\
\bottomrule
\end{tabular}}
\caption{Daily averages for cross-chain arbitrage fee, trade count, and volume before and after the Dencun upgrade (March 13, 2024), with mean change $\Delta$, 95\% \glspl{ci}, Welch $t$, $p$, and Cohen’s $d$.}
\label{tab:activity_dencun_change}
\end{table}

%% file: tables/arb_stats.tex
\begin{table}[b!]
\centering
\footnotesize
\begin{tabular}{lrrrrrrr}
\toprule
\textbf{Chain-Pair} 
& \multicolumn{2}{c}{\textbf{Inventory}} 
& \multicolumn{2}{c}{\textbf{Multi-chain}} 
& \multicolumn{2}{c}{\textbf{Native}} \\
\cmidrule(lr){2-3} \cmidrule(lr){4-5} \cmidrule(lr){6-7}
Category
& Count & \%
& Count & \%
& Count & \% \\
\midrule
L1--L1 & \SI{12533}{} & \SI{86.52}{\percent}
       & \SI{1953}{}  & \SI{13.48}{\percent} 
       & \SI{0}{}     & \SI{0.00}{\percent}  \\
L1--L2 & \SI{93115}{} & \SI{65.79}{\percent}
       & \SI{25948}{} & \SI{18.33}{\percent} 
       & \SI{22464}{} & \SI{15.87}{\percent} \\
L2--L2 & \SI{56755}{} & \SI{65.59}{\percent} 
       & \SI{29725}{} & \SI{34.36}{\percent}  
       & \SI{42}{}    & \SI{0.05}{\percent} \\
\bottomrule
\end{tabular}
\caption{Distribution of cross‐chain arbitrages by chain-pair category and execution method.}
\label{tab:chain_exec_types}
\end{table}

%% file: tables/sia_example.tex
\begin{table}[hbtp]
\centering
\setlength{\tabcolsep}{2pt}
\begin{adjustbox}{width=\columnwidth}
\begin{tabular}{lllllr}
\toprule
\textbf{Transaction Hash} & \textbf{Blockchain} & \textbf{From} & \textbf{To} & \textbf{Type} & \textbf{Time} \\
\midrule
\href{https://basescan.org/tx/0xeb546c545453b6b1ecc6ebf8ad694d4b4dafcc208c330cc505f10b4954037f7c}{\text{0xeb546...37f7c}}
 & Base & \href{https://basescan.org/address/0x4cb6f0ef0eeb503f8065af1a6e6d5dd46197d3d9}{0x4cb...3d9} & \href{https://basescan.org/address/0x6bded42c6da8fbf0d2ba55b2fa120c5e0c8d7891}{0x6bd...891} & Swap & 12:25:51 AM\\
\href{https://arbiscan.io/tx/0xffefbd860c0b9a3219c8961a5566b995ed29286683c146fcc13851856c3aa7aa}{\text{0xffefb...aa7aa}}
 & Arbitrum & \href{https://arbiscan.io/address/0x4cb6f0ef0eeb503f8065af1a6e6d5dd46197d3d9}{0x4cb...3d9} & \href{https://arbiscan.io/address/0x1b02da8cb0d097eb8d57a175b88c7d8b47997506}{0x1b0...506} & Swap & 12:26:18 AM\\
\bottomrule
\end{tabular}
\end{adjustbox}
\caption{Breakdown of an inventory arbitrage between Base and Arbitrum.}
\label{tab:sia}
\end{table}

%% file: tables/sda_example.tex
\begin{table}[hbtp]
\centering
\setlength{\tabcolsep}{2pt}
\begin{adjustbox}{width=\columnwidth}

\begin{tabular}{lllllr}
\toprule
\textbf{Transaction Hash} & \textbf{Blockchain} & \textbf{From} & \textbf{To} & \textbf{Type} & \textbf{Time} \\
\midrule
\href{https://etherscan.io/tx/0x59237caf9f0a47e302e071735df3d4a8835f20d47d00d283967340eb91e61999}{\text{0x59237...61999}} & Ethereum & \href{https://etherscan.io/address/0xe6b65f2e6d8e0c15425bf55fd1fdcdb4b33c3532}{0xe6b...532} & \href{https://etherscan.io/address/0x111111125421ca6dc452d289314280a0f8842a65}{0x111...a65} & Swap & 11:13:11 AM\\
\href{https://etherscan.io/tx/0x64fa3c1d06a94d00ecc4b66fe58acda164f7513c0bc7e5150a79df1eac546aa8}{\text{0x64fa3...46aa8}} & Ethereum & \href{https://etherscan.io/address/0xe6b65f2e6d8e0c15425bf55fd1fdcdb4b33c3532}{0xe6b...532} & \href{https://etherscan.io/address/0xa7c14010afa616fa23a2bb0a94d76dd57dde644d}{0xa7c...44d} & Bridge & 11:14:11 AM\\
\href{https://bscscan.com/tx/0x7330712b229db29174df1d066b0b80eda2814bbfa2e5a54830eb25f0801591ec}{\text{0x73307...591ec}} & BSC$^{\ast}$ & \href{https://bscscan.com/address/0xe6b65f2e6d8e0c15425bf55fd1fdcdb4b33c3532}{0xe6b...532} & \href{https://bscscan.com/address/0xa7c14010afa616fa23a2bb0a94d76dd57dde644d}{0xa7c...44d} & Bridge & 11:19:10 AM\\
\href{https://bscscan.com/tx/0x101d01bd93dde6eb66d177425bc3d04761c9c3ab63f1583dbd7856b5dbccea91}{\text{0x101d0...cea91}}
 & BSC$^{\ast}$ & \href{https://bscscan.com/address/0xe6b65f2e6d8e0c15425bf55fd1fdcdb4b33c3532}{0xe6b...532} & \href{https://bscscan.com/address/0x111111125421ca6dc452d289314280a0f8842a65}{0x111...a65} & Swap & 11:19:22 AM \\  
\bottomrule
$^{\ast}$ Binance Smart Chain.
\end{tabular}
\end{adjustbox}
\caption{Breakdown of a bridge arbitrage between Ethereum and BSC.}
\label{tab:sda}

\end{table}

%% file: tables/chain_pairs_stats_formatted.tex
\begin{table*}[t!]
\centering
\footnotesize
\setlength{\tabcolsep}{4pt}
\begin{tabular}{
  lrrrrrrrrrrr
}
\toprule
\textbf{Chain Pair}
  & \multicolumn{2}{c}{\textbf{Total [USD]}}
  & \multicolumn{2}{c}{\textbf{Average [USD]}}
  & \textbf{Count}
  & \multicolumn{3}{c}{\textbf{Execution Method [\%]}}
  & \multicolumn{3}{c}{\textbf{Settlement Time [s]}} \\
\cmidrule(lr){2-3} \cmidrule(lr){4-5} \cmidrule(lr){7-9} \cmidrule(lr){10-12}
  & Volume   & Profit
  & Volume   & Profit
  & 
  & Inventory      & Multi-chain  & Native &
25th & 50th & 75th \\

\midrule
ARB-ETH & \SI{243611503.97}{} & \SI{1640759.86}{} & \SI{11757.31}{} & \SI{79.19}{} & \SI{20720}{} & \SI{50.77}{} & \SI{29.05}{} & \SI{20.17}{} & \SI{12.00}{} & \SI{105.00}{} & \SI{369.00}{} \\
BASE-ETH & \SI{136809634.57}{} & \SI{1671060.39}{} & \SI{6709.64}{} & \SI{81.95}{} & \SI{20390}{} & \SI{54.70}{} & \SI{20.76}{} & \SI{24.55}{} & \SI{14.00}{} & \SI{194.00}{} & \SI{318.00}{} \\
BSC-ETH & \SI{106614309.82}{} & \SI{1447488.64}{} & \SI{7359.82}{} & \SI{99.92}{} & \SI{14486}{} & \SI{86.52}{} & \SI{13.48}{} & \SI{0.00}{} & \SI{9.00}{} & \SI{97.00}{} & \SI{285.00}{} \\
ETH-POL & \SI{103521335.22}{} & \SI{1473049.96}{} & \SI{3760.17}{} & \SI{53.51}{} & \SI{27531}{} & \SI{56.41}{} & \SI{2.50}{} & \SI{41.10}{} & \SI{188.50}{} & \SI{1251.00}{} & \SI{1888.00}{} \\
ARB-BASE & \SI{43308426.15}{} & \SI{276184.86}{} & \SI{2556.43}{} & \SI{16.30}{} & \SI{16941}{} & \SI{68.18}{} & \SI{31.82}{} & \SI{0.00}{} & \SI{2.00}{} & \SI{28.00}{} & \SI{130.00}{} \\
BASE-OP & \SI{42010397.87}{} & \SI{232645.08}{} & \SI{2375.75}{} & \SI{13.16}{} & \SI{17683}{} & \SI{57.44}{} & \SI{42.54}{} & \SI{0.02}{} & \SI{10.00}{} & \SI{66.00}{} & \SI{218.00}{} \\
ARB-BSC & \SI{34239480.03}{} & \SI{303282.66}{} & \SI{1331.81}{} & \SI{11.80}{} & \SI{25709}{} & \SI{65.55}{} & \SI{34.45}{} & \SI{0.00}{} & \SI{2.00}{} & \SI{7.00}{} & \SI{74.00}{} \\
ARB-POL & \SI{31883971.14}{} & \SI{167111.54}{} & \SI{2310.77}{} & \SI{12.11}{} & \SI{13798}{} & \SI{82.54}{} & \SI{17.44}{} & \SI{0.01}{} & \SI{5.00}{} & \SI{8.00}{} & \SI{190.00}{} \\
ARB-OP & \SI{25277935.69}{} & \SI{133137.93}{} & \SI{2019.65}{} & \SI{10.64}{} & \SI{12516}{} & \SI{60.17}{} & \SI{39.83}{} & \SI{0.00}{} & \SI{19.00}{} & \SI{89.00}{} & \SI{403.00}{} \\
AVA-ETH & \SI{15546710.75}{} & \SI{174235.31}{} & \SI{8152.44}{} & \SI{91.37}{} & \SI{1907}{} & \SI{83.69}{} & \SI{16.31}{} & \SI{0.00}{} & \SI{9.00}{} & \SI{17.00}{} & \SI{247.50}{} \\
ETH-OP & \SI{13165677.85}{} & \SI{233166.07}{} & \SI{4246.99}{} & \SI{75.21}{} & \SI{3100}{} & \SI{28.48}{} & \SI{8.52}{} & \SI{63.00}{} & \SI{84.00}{} & \SI{102.00}{} & \SI{154.00}{} \\
ARB-SCROLL & \SI{12323983.83}{} & \SI{8504.15}{} & \SI{10877.30}{} & \SI{7.51}{} & \SI{1133}{} & \SI{99.47}{} & \SI{0.53}{} & \SI{0.00}{} & \SI{3.00}{} & \SI{5.00}{} & \SI{18.00}{} \\
ARB-AVA & \SI{11337989.28}{} & \SI{299109.72}{} & \SI{1483.25}{} & \SI{39.13}{} & \SI{7644}{} & \SI{42.61}{} & \SI{57.39}{} & \SI{0.00}{} & \SI{9.00}{} & \SI{38.00}{} & \SI{81.00}{} \\
BASE-POL & \SI{9875258.02}{} & \SI{78817.79}{} & \SI{1316.17}{} & \SI{10.50}{} & \SI{7503}{} & \SI{81.01}{} & \SI{18.75}{} & \SI{0.24}{} & \SI{6.00}{} & \SI{117.00}{} & \SI{416.50}{} \\
BSC-POL & \SI{9109320.4}{} & \SI{183862.16}{} & \SI{376.82}{} & \SI{7.61}{} & \SI{24174}{} & \SI{93.34}{} & \SI{6.65}{} & \SI{0.00}{} & \SI{4.00}{} & \SI{6.00}{} & \SI{61.00}{} \\
BASE-BSC & \SI{8374722.52}{} & \SI{94455.39}{} & \SI{1024.06}{} & \SI{11.55}{} & \SI{8178}{} & \SI{72.57}{} & \SI{27.43}{} & \SI{0.00}{} & \SI{3.00}{} & \SI{53.00}{} & \SI{201.75}{} \\
AVA-BSC & \SI{6586964.25}{} & \SI{74915.78}{} & \SI{792.66}{} & \SI{9.02}{} & \SI{8310}{} & \SI{86.75}{} & \SI{13.25}{} & \SI{0.00}{} & \SI{3.00}{} & \SI{5.00}{} & \SI{47.00}{} \\
ETH-SCROLL & \SI{3674739.91}{} & \SI{3374.46}{} & \SI{34343.36}{} & \SI{31.54}{} & \SI{107}{} & \SI{71.03}{} & \SI{24.30}{} & \SI{4.67}{} & \SI{20.50}{} & \SI{27.00}{} & \SI{369.50}{} \\
OP-POL & \SI{3424785.64}{} & \SI{21126.36}{} & \SI{1105.48}{} & \SI{6.82}{} & \SI{3098}{} & \SI{45.80}{} & \SI{53.58}{} & \SI{0.61}{} & \SI{168.00}{} & \SI{382.00}{} & \SI{807.75}{} \\
AVA-POL & \SI{2111453.94}{} & \SI{35268.49}{} & \SI{693.42}{} & \SI{11.58}{} & \SI{3045}{} & \SI{84.60}{} & \SI{15.40}{} & \SI{0.00}{} & \SI{7.00}{} & \SI{41.00}{} & \SI{264.00}{} \\
BSC-OP & \SI{1530040.17}{} & \SI{54319.71}{} & \SI{1289.00}{} & \SI{45.76}{} & \SI{1187}{} & \SI{57.20}{} & \SI{42.80}{} & \SI{0.00}{} & \SI{7.00}{} & \SI{105.00}{} & \SI{642.00}{} \\
OP-SCROLL & \SI{1166887.32}{} & \SI{1450.7}{} & \SI{10705.39}{} & \SI{13.31}{} & \SI{109}{} & \SI{100.00}{} & \SI{0.00}{} & \SI{0.00}{} & \SI{3.00}{} & \SI{7.00}{} & \SI{20.00}{} \\
AVA-BASE & \SI{1065823.0}{} & \SI{20944.82}{} & \SI{666.97}{} & \SI{13.11}{} & \SI{1598}{} & \SI{61.45}{} & \SI{38.55}{} & \SI{0.00}{} & \SI{6.00}{} & \SI{46.00}{} & \SI{113.00}{} \\
AVA-OP & \SI{540371.31}{} & \SI{7904.1}{} & \SI{1382.02}{} & \SI{20.22}{} & \SI{391}{} & \SI{40.15}{} & \SI{59.85}{} & \SI{0.00}{} & \SI{58.50}{} & \SI{138.00}{} & \SI{562.00}{} \\
ARB-ZKSYNC & \SI{456054.82}{} & \SI{1193.76}{} & \SI{1646.41}{} & \SI{4.31}{} & \SI{277}{} & \SI{36.82}{} & \SI{63.18}{} & \SI{0.00}{} & \SI{23.00}{} & \SI{176.00}{} & \SI{380.00}{} \\
POL-SCROLL & \SI{429752.1}{} & \SI{779.99}{} & \SI{9767.09}{} & \SI{17.73}{} & \SI{44}{} & \SI{97.73}{} & \SI{2.27}{} & \SI{0.00}{} & \SI{15.00}{} & \SI{26.00}{} & \SI{283.50}{} \\
ETH-ZKSYNC & \SI{138119.84}{} & \SI{164.84}{} & \SI{6005.21}{} & \SI{7.17}{} & \SI{23}{} & \SI{30.43}{} & \SI{43.48}{} & \SI{26.09}{} & \SI{353.00}{} & \SI{657.00}{} & \SI{1259.50}{} \\
BSC-ZKSYNC & \SI{129435.05}{} & \SI{8571.36}{} & \SI{1307.42}{} & \SI{86.58}{} & \SI{99}{} & \SI{31.31}{} & \SI{68.69}{} & \SI{0.00}{} & \SI{108.00}{} & \SI{154.00}{} & \SI{383.50}{} \\
BASE-SCROLL & \SI{116708.87}{} & \SI{246.81}{} & \SI{2714.16}{} & \SI{5.74}{} & \SI{43}{} & \SI{74.42}{} & \SI{25.58}{} & \SI{0.00}{} & \SI{18.00}{} & \SI{26.00}{} & \SI{102.50}{} \\
BASE-ZKSYNC & \SI{115157.55}{} & \SI{3615.37}{} & \SI{170.10}{} & \SI{5.34}{} & \SI{677}{} & \SI{33.83}{} & \SI{66.17}{} & \SI{0.00}{} & \SI{154.00}{} & \SI{301.00}{} & \SI{575.00}{} \\
BSC-SCROLL & \SI{71601.75}{} & \SI{1053.71}{} & \SI{778.28}{} & \SI{11.45}{} & \SI{92}{} & \SI{84.78}{} & \SI{15.22}{} & \SI{0.00}{} & \SI{121.00}{} & \SI{172.00}{} & \SI{213.75}{} \\
OP-ZKSYNC & \SI{58086.12}{} & \SI{1.11}{} & \SI{3227.01}{} & \SI{0.06}{} & \SI{18}{} & \SI{66.67}{} & \SI{33.33}{} & \SI{0.00}{} & \SI{10.00}{} & \SI{17.50}{} & \SI{171.50}{} \\
SCROLL-ZKSYNC & \SI{15231.91}{} & \SI{-2.73}{} & \SI{3807.98}{} & \SI{-0.68}{} & \SI{4}{} & \SI{100.00}{} & \SI{0.00}{} & \SI{0.00}{} & \SI{11.50}{} & \SI{12.00}{} & \SI{14.75}{} \\
\bottomrule
\end{tabular}
\caption{Cross-chain arbitrage statistics for all chain-pairs in our dataset.}
\label{tab:all_chain_pairs_metrics}
\end{table*}

%% file: tables/top_token_pairs_stats.tex
\begin{table}[hbtp]
\centering
\footnotesize
\begin{tabular}{lrrrrr}
\toprule
\textbf{Token Pair} & \multicolumn{2}{c}{\textbf{Total [USD]}} & \multicolumn{2}{c}{\textbf{Average [USD]}} & \textbf{Count} \\
\cmidrule(lr){2-3} \cmidrule(lr){4-5}
& Volume & Profit & Volume & Profit & \\
\midrule
BTC-ETH & \SI{94481876.03}{} & \SI{115151.06}{} & \SI{31368.48}{} & \SI{38.23}{} & \SI{3012}{} \\
MIM-USD & \SI{53992514.62}{} & \SI{113038.14}{} & \SI{43263.23}{} & \SI{90.58}{} & \SI{1248}{} \\
NMT-USD & \SI{47672408.90}{} & \SI{486996.21}{} & \SI{28822.50}{} & \SI{294.44}{} & \SI{1654}{} \\
DMT-ETH & \SI{46490116.88}{} & \SI{306793.02}{} & \SI{12269.76}{} & \SI{80.97}{} & \SI{3789}{} \\
ETH-OMNI & \SI{32253198.96}{} & \SI{403176.99}{} & \SI{2627.98}{} & \SI{32.85}{} & \SI{12273}{} \\
ETH-PRIME & \SI{30026851.42}{} & \SI{160867.88}{} & \SI{19334.74}{} & \SI{103.59}{} & \SI{1553}{} \\
USD-agEUR & \SI{19521861.60}{} & \SI{219876.53}{} & \SI{5894.28}{} & \SI{66.39}{} & \SI{3312}{} \\
ETH-NPC & \SI{18818700.27}{} & \SI{147817.12}{} & \SI{5162.88}{} & \SI{40.55}{} & \SI{3645}{} \\
SHRAP-USD & \SI{15426515.75}{} & \SI{112919.70}{} & \SI{5352.71}{} & \SI{39.18}{} & \SI{2882}{} \\
ETH-USD & \SI{15223335.86}{} & \SI{27860.14}{} & \SI{5137.81}{} & \SI{9.40}{} & \SI{2963}{} \\
\bottomrule
\end{tabular}
\caption{Summary statistics for the ten token pairs with the highest cumulative arbitrage volume.}
\label{tab:top_token_pairs_stats}
\end{table}

%% file: tables/bridge_protocol_stats.tex
\begin{table}[hbtp]
\centering
\footnotesize
\begin{tabular}{lrrrr}
\toprule
\textbf{Bridge} &
\textbf{Count} &
\multicolumn{2}{c}{\textbf{Averages}} &
\textbf{Cost Ratio [\%]} \\
\cmidrule(lr){3-4}
&  &
Duration [s] & Cost [USD] &
Bridge Fee/Tx Fee \\
\midrule

LayerZero & \SI{39536}{} & \SI{165.56}{} & \SI{1.64}{} & \SI{26.86}{} \\
Native & \SI{22506}{} & \SI{739.40}{} & \SI{4.91}{} & \SI{30.29}{} \\
Wormhole & \SI{8566}{} & \SI{990.68}{} & \SI{2.03}{} & \SI{29.63}{} \\
Axelar & \SI{3212}{} & \SI{1364.93}{} & \SI{0.45}{} & \SI{21.51}{} \\
Celer & \SI{2815}{} & \SI{281.21}{} & \SI{0.66}{} & \SI{24.55}{} \\
Across & \SI{2639}{} & \SI{269.99}{} & \SI{0.24}{} & \SI{18.83}{} \\
Synapse & \SI{472}{} & \SI{391.28}{} & \SI{0.94}{} & \SI{22.59}{} \\
Stargate & \SI{315}{} & \SI{133.66}{} & \SI{0.93}{} & \SI{26.07}{} \\
\bottomrule
\end{tabular}
\caption{Summary statistics for bridge protocols used in $\geq$ \SI{0.1}{\%} of bridge-based cross-chain arbitrages.}
\label{tab:bridge_protocol_stats}
\end{table}

%% file: tables/top_searcher_dencun.tex
\begin{table}[b!]
\centering
\footnotesize     
\resizebox{\columnwidth}{!}{%
\begin{tabular}{lrrrrrrrr}
\toprule
\textbf{Metric} & \textbf{Pre} & \textbf{Post} & $\Delta$ & \textbf{95\% CI} & $t$ & $df$ & $p$ & $d$ \\
\midrule
Fee (USD)    & 10 & 6.57 & -3.44 & $[-4.68,-2.19]$ & 5.42 & 364 & 1.08e-7 & -0.56 \\
Count        & 79.4 & 175 & 95.4 & $[84.5,106]$ &  -17.3 & 341 & 9.44e-49 & 1.78 \\
Volume (USD) & 3.57e5 & 1.3e6 & 9.38e5 & $[0.8,1.07]\!\times\!10^{6}$ & -14 & 238 & 6.52e-33 & 1.52 \\
Volume (\%) & 20.3\% & 39.7\% & 19.3\% & $[17.3, 21.4]\%$ & -18.6 & 340 & 8.14e-54 & 1.96 \\
\bottomrule
\end{tabular}}
\caption{Arbitrageur \texttt{0xCA74}'s daily averages for cross-chain arbitrage fee, trade count, and volume before and after the Dencun upgrade (March 13, 2024), with mean change $\Delta$, 95\% \glspl{ci}, Welch $t$, $p$, and Cohen’s $d$.}
\label{tab:top_searcher_dencun_change}
\end{table}

%% file: tables/searcher_stats_formatted.tex
\begin{table*}[t!]
\centering
\footnotesize
\begin{tabular}{llrrrrrrrrrr}
\toprule
\multicolumn{2}{c}{\textbf{Arbitrageur}} &
\multicolumn{2}{c}{\textbf{Totals [USD]}} &
\multicolumn{2}{c}{\textbf{Averages [USD]}} &
\textbf{Count} &
\multicolumn{2}{c}{\textbf{Execution Method [\%]}} &
\multicolumn{3}{c}{\textbf{Settlement Time [s]}} \\
\cmidrule(ll){1-2}
\cmidrule(lr){3-4}
\cmidrule(lr){5-6}
\cmidrule(lr){8-9}
\cmidrule(lr){10-12}
Address & Type &
Volume & Profit & 
Volume & Profit &
&
Inventory & Bridge &
25th & 50th & 75th \\
\midrule
\href{https://etherscan.io/address/0xca74f404e0c7bfa35b13b511097df966d5a65597}{0xCA74} & EOA & \SI{293534534.50}{} & \SI{2043570.41}{} & \SI{6433.77}{} & \SI{44.79}{} & \SI{45624}{} & \SI{68.00}{} & \SI{32.00}{} & \SI{6.00}{} & \SI{64.00}{} & \SI{528.00}{} \\
\href{https://etherscan.io/address/0xa311f7ca3eb2fb98123a807f9b8e4bbbdbdcb2ee}{0xA311}$^{\dagger}$ & SC & \SI{64089695.72}{} & \SI{578432.65}{} & \SI{5156.46}{} & \SI{46.54}{} & \SI{12429}{} & \SI{50.90}{} & \SI{49.10}{} & \SI{3.00}{} & \SI{27.00}{} & \SI{58.00}{} \\
\href{https://etherscan.io/address/0x4cb6f0ef0eeb503f8065af1a6e6d5dd46197d3d9}{0x4CB6} & EOA & \SI{59428683.84}{} & \SI{400555.17}{} & \SI{2745.10}{} & \SI{18.50}{} & \SI{21649}{} & \SI{73.90}{} & \SI{26.10}{} & \SI{2.00}{} & \SI{6.00}{} & \SI{38.00}{} \\
\href{https://etherscan.io/address/0x0a6c69327d517568e6308f1e1cd2fd2b2b3cd4bf}{0x0A6C}$^{\dagger}$ & SC & \SI{26573871.45}{} & \SI{282973.40}{} & \SI{2414.49}{} & \SI{25.71}{} & \SI{11006}{} & \SI{21.69}{} & \SI{78.31}{} & \SI{134.00}{} & \SI{228.00}{} & \SI{575.00}{} \\
\href{https://etherscan.io/address/0x9b9019d7faf9857d9643e65a15327a755042884b}{0x9B90} & EOA & \SI{21928158.44}{} & \SI{241282.77}{} & \SI{14407.46}{} & \SI{158.53}{} & \SI{1522}{} & \SI{64.91}{} & \SI{35.09}{} & \SI{9.00}{} & \SI{13.00}{} & \SI{60.00}{} \\
\href{https://etherscan.io/address/0x55cd38a871dbc474b07d7f6a0a1fe0eabd3155cd}{0x55CD} & EOA & \SI{21378753.99}{} & \SI{12090.77}{} & \SI{194352.31}{} & \SI{109.92}{} & \SI{110}{} & \SI{19.09}{} & \SI{80.91}{} & \SI{89.00}{} & \SI{168.00}{} & \SI{281.75}{} \\
\href{https://etherscan.io/address/0x612ed65a88bfa101e6c44727919b6f81df92e428}{0x612E} & EOA & \SI{14721732.01}{} & \SI{101750.09}{} & \SI{23554.77}{} & \SI{162.80}{} & \SI{625}{} & \SI{100.00}{} & \SI{0.00}{} & \SI{79.00}{} & \SI{117.00}{} & \SI{138.00}{} \\
\href{https://etherscan.io/address/0x48cce57c4d2dbb31eaf79575abf482bbb8dc071d}{0x48CC}  & EOA &  \SI{14495989.50}{} & \SI{10637.47}{} & \SI{35270.05}{} & \SI{25.88}{} & \SI{411}{} & \SI{100.00}{} & \SI{0.00}{} & \SI{17.00}{} & \SI{21.00}{} & \SI{25.00}{} \\
\href{https://etherscan.io/address/0x882d04c3d8410ddf2061b3cba2c3522854316feb}{0x882D} & SC &  \SI{11252232.57}{} & \SI{189184.18}{} & \SI{3436.85}{} & \SI{57.78}{} & \SI{3274}{} & \SI{53.12}{} & \SI{46.88}{} & \SI{1138.00}{} & \SI{1345.00}{} & \SI{2299.00}{} \\
\href{https://etherscan.io/address/0xe83f75907fb4c575414fa6f5cfe8cef24dc5870c}{0xE83F} & EOA & \SI{10276396.37}{} & \SI{81094.31}{} & \SI{4336.03}{} & \SI{34.22}{} & \SI{2370}{} & \SI{45.02}{} & \SI{54.98}{} & \SI{82.00}{} & \SI{232.00}{} & \SI{936.25}{} \\
\href{https://etherscan.io/address/0x622661ab4b6ab93c659e751f47ebb0c6e6ad9f48}{0x6226} & EOA & \SI{9664552.84}{} & \SI{141303.17}{} & \SI{246.36}{} & \SI{3.60}{} & \SI{39229}{} & \SI{99.99}{} & \SI{0.01}{} & \SI{2.00}{} & \SI{4.00}{} & \SI{6.00}{} \\
\bottomrule
\multicolumn{4}{l}{$^{\dagger}$EOA identified via its SC.}
\end{tabular}
\caption{Cross-chain arbitrage metrics for arbitrageurs with $\geq$\SI{1}{\%} of the total trade volume. EOA - Externally Owned Account, SC - Smart Contract.}
\label{tab:arber_metrics}
\end{table*}